\documentclass[11pt]{article}

\usepackage{color}
\usepackage{epsfig}
\usepackage{epstopdf}
\usepackage[section]{placeins}
\usepackage{graphicx}
\usepackage{amsmath}
\usepackage{amssymb}
\usepackage{amsfonts}
\usepackage{url}
\usepackage{multirow}
\usepackage{array}
\usepackage{latexsym}
\usepackage{natbib}
\usepackage[noline,ruled]{algorithm2e}
\usepackage{bigstrut}
\usepackage{subcaption}
\usepackage{dutchcal}
\usepackage{kotex}

\usepackage[toc,page]{appendix}
\usepackage{comment}
\usepackage{kotex}

\newsavebox{\theorembox}
\newsavebox{\lemmabox}
\newsavebox{\claimbox}
\newsavebox{\factbox}
\newsavebox{\corollarybox}
\newsavebox{\examplebox}
\newsavebox{\remarkbox}
\newsavebox{\assbox}
\newsavebox{\propositionbox}
\newsavebox{\problembox}
\newsavebox{\defbox}

\savebox{\theorembox}{\noindent\bf Theorem}
\savebox{\lemmabox}{\noindent\bf Lemma}
\savebox{\factbox}{\noindent\bf Fact}
\savebox{\corollarybox}{\noindent\bf Corollary}
\savebox{\examplebox}{\noindent\bf Example}
\savebox{\remarkbox}{\noindent\bf Remark}
\savebox{\assbox}{\noindent\bf Assumption}
\savebox{\propositionbox}{\noindent\bf Proposition}
\savebox{\problembox}{\noindent\bf Problem}
\savebox{\defbox}{\noindent\bf Definition}

\newtheorem{prop}{\usebox{\propositionbox}}
\newtheorem{lem}{\usebox{\lemmabox}}

\def\blackslug{\hbox{\hskip 1pt \vrule width 4pt height 8pt depth 1.5pt
\hskip 1pt}}

\newcommand{\qed}{\mbox{}\hspace*{\fill}\nolinebreak\mbox{$\rule{0.7em}{0.7em}$}}

\newenvironment{proof}[1][\noindent \bf Proof:]{\begin{trivlist}
\item[\hskip \labelsep {\bfseries #1}]}{\end{trivlist}}

\newcommand{\vs}{\vspace*{3mm}}

\newcommand{\lt}{\left}
\newcommand{\rt}{\right}
\newcommand{\bex}{\begin{eqnarray*}}
\newcommand{\eex}{\end{eqnarray*}}

\newcommand{\rd}{{\rm d}}

\def\bb{{\bf b}}

\def\bI{{\bf I}}

\def\bW{{\bf W}}

\def\brho{\mbox{\boldmath$\rho$}}

\def\bbP{\mathbb{P}}

\allowdisplaybreaks

\begin{document}

%
%
%
%
%

\title{Strategic Users in a Priority Queue with Bulk Service on Blockchains}

\author{
Donghwa Seo\thanks{Samsung SDS, Seoul, South Korea, E-mail: {\tt properitas95@gmail.com}} \\
\small{\it Samsung SDS}
\vspace{3mm}\\
Kyoung-Kuk Kim\thanks{Corresponding author, KAIST, College of Business, Seoul, South Korea, E-mail: {\tt kkim@kaist.ac.kr}}\\
\small{\it Korea Advanced Institute of Science and Technology}
}

\date{July 2025}

\maketitle

\baselineskip 18pt
\begin{abstract}
This paper analyzes transaction fees on blockchains by considering that they form a priority queue and users play a queueing game. This modeling approach using M/G\textsuperscript{$K$}/1 priority queue, we provide new insights into the dynamics governing transaction fees and its impact on user behavior. This work contributes to the literature in that we find semi-closed form expressions for steady-state quantities in the target queue and that the relationship between the delay cost of a user and the transaction fee (bid) is extended to the case of general block generation time. We apply our results to the Bitcoin network and simulate user responses under various scenarios. Cross-chain analysis across Bitcoin, Dogecoin, and Litecoin reveals similarities in normalized cost structures.

\vs
\noindent
{\sc Keywords:} Queueing; Blockchain; Transaction fees; M/G\textsuperscript{$K$}/1 queue; Prioritized queuing game 
\end{abstract}

\section{Introduction}\label{sec:intro}
The advent of blockchain technology has revolutionized decentralized systems, fostering the development of cryptocurrencies, decentralized exchanges, non-fungible tokens, and various decentralized applications. Public blockchains, in particular, offer distinct advantages: universal access, data integrity, decentralized operations, and transparent data that enhance trust. While such innovations and advantages are the core of a rapidly expanding industry, the focus of academic research has been more on system stability and system optimization. User interactions with blockchain services have been relatively less explored. 

From the users' perspective in a public blockchain system, they experience a unique form of waiting. They broadcast transactions, hoping to be included in the next block. Since multiple transactions are processed in a new block, this procedure is  akin to a queue served in bulk. Unlike fixed-price services, users bid for priority and transactions with higher bids (fees) are chosen earlier than transactions with lower bids. This real-time auction therefore incurs time-varying transaction costs because bids are relatively higher (lower) when trading intensity is higher (lower). Hence, the form of waiting of blockchain users can be thought of as a priority queue with bulk service. The study of blockchain queues but remains sparse. In this paper, we focus on Bitcoin, the representative cryptocurrency, and analyze user interactions with the blockchain via a queueing game approach. We aim to provide a better understanding of the transaction fee dynamics and users' waiting costs, ultimately helping to resolve network congestion and to enhance network throughput. Our work was motivated by the emergence of BRC-20 in early 2023 which laid the technical foundation for non-fungible tokens to function on the Bitcoin network, enabling the development of associated applications. This expansion of functionality has led to increased network congestion and longer transaction queues. 



Our modeling approach is to see the blockchain queue as M/G\textsuperscript{$K$}/1 queue. Transaction arrivals follow a Poisson process, block generation times have a general distribution with $K$ transactions served in each block, and there is a single server. Based on the assumption that users make optimal bids associated with their waiting costs, we infer users' cost structure from transaction data. For this, we extend the existing study on bulk queues and provide a new method of calculating expected waiting times and expected queue length. We believe that our results give new insights into users' cost structure and strategic bidding behaviors. This also allows us to examine possible consequences of protocol changes or strategic actions of miners on user behaviors. For this purpose, we numerically test the effects of changes in the distribution of block generation time. Before we present our model in detail, the relevant literature is reviewed in the following paragraphs. 


\noindent 
{\bf Queueing approach.} Bulk queues have been researched for many years, starting in the mid 20th century. \cite{bailey1954queueing} studies the waiting time distribution of patients at a hospital where there is a maximum number that a medical consultant can see at a session. The service time distribution is modeled as $\chi^2$ distribution. \cite{chaudhry1983first} provide a comprehensive treatment of bulk queues, and their presentation of M/G\textsuperscript{$K$}/1 is highly relevant to our model setup. There are more recent developments on bulk queues such as \cite{claeys2013tail} or \cite{banerjee2015analysis}. Both papers look at batch arrivals and bulk services. The former focuses on the tail probabilities of customer waiting whereas the latter is on steady state probabilities particularly when arrivals are Markovian and service times are of phase type.   

Such earlier papers on bulk queues have been applied to blockchains in recent years. \cite{kawase2017transaction} and \cite{kasahara2019effect} present initial attempts to analyze transaction-confirmation time for Bitcoin and the impact of transaction fees. The authors use this same bulk queue model with Poisson arrivals and general service distribution. However, as pointed out in \cite{li2018blockchain}, their approach involves infinitely many unknown numbers so that actual implementation is restricted to the case of exponential service time. In order to circumvent this problem, \cite{li2018blockchain} use a continuous time Markov process of GI/M/1 type. On the other hand, Queueing approach has been applied not only to transaction queues but also to blocks generated at blockchain nodes. \cite{papadis2018stochastic} develop a stochastic model for such queues of minted blocks to study the impact of block dissemination delays.

\noindent 
{\bf Transaction fees.} The Bitcoin blockchain system rewards miners who successfully generate a new block with a fixed amount of Bitcoins. The total supply of Bitcoins is capped, and in the long run, the system will rely on transaction fees.  For this reason, various aspects of fee dynamics and their implications have received increasing attention as the size of the network grows. \cite{houy2014economics} is an early work on transaction fees where the author considers a simple partial equilibrium model and a fee is viewed as the price for the block space. \cite{moser2015trends} offer empirical aspects of transaction fees. They find that higher fees indeed result in faster confirmation and that impatient users offer higher fees. A more recent empirical work is \cite{ilk2021stability}. This work specifies the inelastic nature  of the demand curve of users whereas   the supply (of block space) curve of miners is elastic to transaction fees. The recent literature on fees such as \cite{lavi2022redesign, basu2023stable, rough2024fee} reconsider transaction fee mechanisms. 

\noindent
{\bf Game theoretic approach.} To understand user behaviors on blockchains with transaction fees, game theoretic models have been employed to account for the prioritized mechanism in place. In fact, the strategic behavior of customers in queueing systems has been extensively studied in the context of congestion games or queueing games. \cite{hassin2003queue} provide a comprehensive overview, detailing various models such as heterogeneous customers, competing servers, queues with reneging, and queues with priorities. More recent works include \cite{dimitri2019transaction, easley2019mining}. In \cite{dimitri2019transaction}, transaction fees are framed as a Nash equilibrium outcome  of a congestion game, where the winning Bitcoin miner (who solves the puzzle) functions as an auctioneer. In \cite{easley2019mining}, users play the game of network participation and fee payment. By assuming that transactions flow in and flow out at fixed rates, the authors derive an equilibrium behavior of network participants and study the impact of  microstructural features.

\cite{huberman2021monopoly} is most relevant to our work. In their model, the transaction queue is represented as an M/M\textsuperscript{$K$}/1 queue with fixed rates. This simplification gives us a closed form expression for the steady-state queue length distribution and expected waiting time. Here, the waiting time is computed for each user with delay cost or load parameter $\rho$. Based on this, users' individual optimization leads to an equilibrium bidding strategy and to an equilibrium revenue for miners. We extend this framework to account for a general service time distribution. To do so, we develop a numerical procedure for computing the steady state distribution, addressing the issue identified by \cite{li2018blockchain} mentioned above. Furthermore, we adopt the approach of \cite{huberman2021monopoly} to infer the optimal bidding strategy and the hidden cost structure of users from transaction data. Our generalized model allows for comparisons of user behavior across different blockchains.

This paper is structured as follows. Section~\ref{sec:model} presents our model setup. In Section~\ref{sec:analysis}, we conduct the steady-state analysis of the transaction queue and compute relevant quantities such as expected queue length. Section~\ref{sec:bidding} studies the optimal bidding strategy of individual users as a function of their delay costs so that we build a link between delay costs and transaction fees. In the section that follows, we consider the Bitcoin transaction data to infer users' cost structure. Section~\ref{sec:conclusion} concludes the paper.

\section{The Model}\label{sec:model}

In a blockchain system with proof-of-work protocol, transaction requests are stored in the mempool, part of which are confirmed by the successful miner who generates a new block. Since users are impatient, their bids for transaction confirmation are positively associated with delay costs. This delay or waiting time of a user then depends on the level of congestion and the distribution of block generation time $B$. See Figure~\ref{fig:3.1.QueuingModel}. Therefore, the final transaction fee collected by the winning miner is the outcome of individual user's optimization based on the load parameter $\rho$, waiting cost $c_i$ with common distribution $F$, and block distribution $F_B$.

\begin{figure}[htb!]
    \centering
    \includegraphics[width=\textwidth]{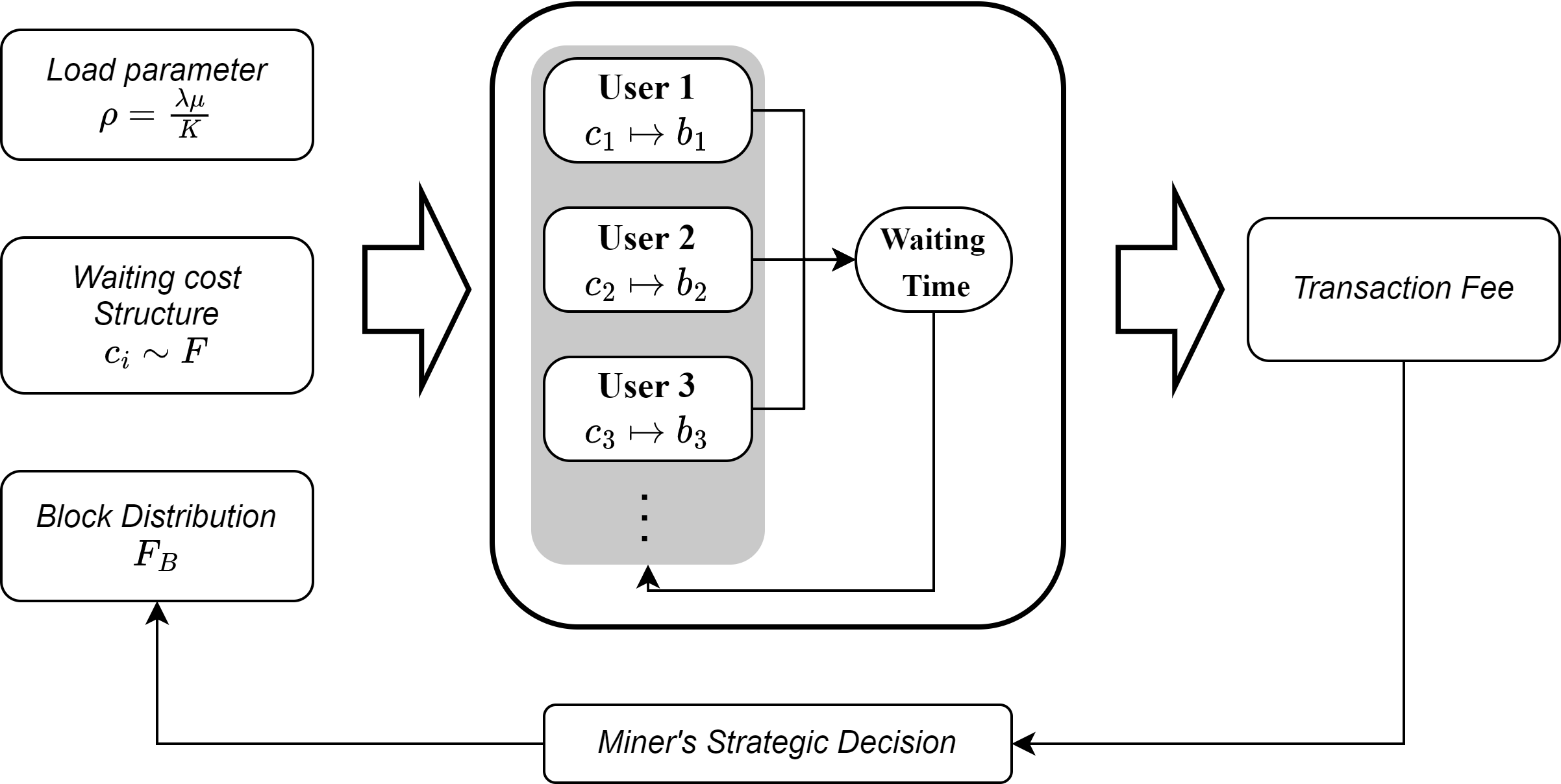}
    \caption{Description of the system components and strategic decisions of users.}
    \label{fig:3.1.QueuingModel}
\end{figure}

We note that the block distribution $F_B$ is also an outcome of collective operating decisions by miners. Since miners are motivated by economic rewards, they weigh rewards against costs and their strategic actions affect $F_B$. For instance, \cite{kim2024mind} study a Nash equilibrium of miners' actions and conclude that there is a possible mining gap, i.e. partial utilization of mining rigs, if the economic reward for mining is not sufficient. In our queueing theoretic analysis, we focus on the user side and treat the block distribution as given. However, we numerically test the impact of distributional changes of block generation time. Lastly, there is yet another important mechanism called difficulty adjustment in the system. This controls the winning probability so as to make actual block generation times close to a certain target, i.e. 10 minutes for Bitcoin. Therefore, the combined effect of miners' operating decisions, users' bidding decisions, and the difficulty adjustment determines the transaction dynamics of a blockchain system. For our analysis, however, we assume fixed block distribution as well as fixed difficulty as we are interested in the steady-state analysis of the queue. 

To concentrate on the system's fundamental dynamics, we make other assumptions as follows. First, transaction requests are uniform in size. Second, each block contains up to $K$ transactions. Third, transaction arrivals follow a Poisson process with a fixed rate $\lambda$. Let us denote the mean block generation time by $\mu = \mathbb{E}[B]$. Then, we can define the load parameter by $\rho = \lambda\mu/K$. Fourth, $\rho <1$ to make the queue stable in the long run. This allows us to see the queue as M/G\textsuperscript{$K$}/1 as long as there is no delay in transaction propagation to the mempool. The queue size dynamics per se does not involve uers' strategic actions as transactions are not treated differently. It is the fee dynamics where users' delay cost structure, block distribution, and the load parameter $\rho$ play together. For this, we make additional assumptions on the user side. 
Users are assumed to be aware of steady-state behaviors but they do not make real-time observations of the queue. Hence, optimal bidding decisions, bid $b_i$ by the $i$-th user, are made by considering the trade-off between the fee payment and the delay cost where this delay is the expected waiting time in the steady-state.  In this setup, we closely follow \cite{huberman2021monopoly}. Specifically, the utility of the $i$-th user is given by $u_i = - b_i - c_i W_i$ where $W_i$ is the expected waiting time for a fee $b_i$. A user exits the system if the utility becomes excessively negative. Section~\ref{sec:bidding} details the optimization of the user utility and finds the relationship between waiting cost  and optimal bid.

\section{Steady State Analysis}\label{sec:analysis}

\subsection{Limiting Distribution} 

This section analyzes a blockchain-inspired bulk service queue where transactions are processed in batches during block generation events. Unlike traditional bulk service queues, our model has no separate service station. Instead, up to $K$ transactions from the mempool are processed instantaneously when a new block is created. Although the selection of transactions depends on transaction fees, the queue dynamics and queue length distributions do not depend on them. 
Hence, in this section, we aim to derive steady state equations for the limiting distribution of the queue length $N(t)$ at time $t$. 

Since the service time distribution or block distribution $F_B$ is general, $N(t)$ itself is not Markovian. However, by considering the time elapsed from the last block generation, say $X(t)$, we have a Markovian random vector $(N(t), X(t))$ with the state space $\{(n,x) | n \in \mathbb{Z}_+, \; x \geq 0\}$. The dynamics of the system is visualized in Figure~\ref{fig:queue_dynamics}. At any point in time, there are three possible cases: no event, single transaction arrival, and  block creation. Three sub-figures show how the state variables $n$ and $x$ change in each case. Figure~\ref{fig:3.4.MG_K1_2} corresponds to the case of no event. The number in the queue stays the same but the time elapses by $\rd t$. On the other hand, if there is a transaction arrival prior to new block creation, then the time elapses by $\rd t$ and the number in the queue increases by 1. Hence, the state variables $n$ and $x$ change simultaneously. Lastly, when there is a new block generated, the elapsed time falls from $x$ to zero and the size of the queue in the mempool decreases by $K$ at maximum.  

\begin{figure}
    \centering
    \begin{subfigure}[htb!]{1\textwidth}
        \centering
        \includegraphics[width=.5\textwidth]{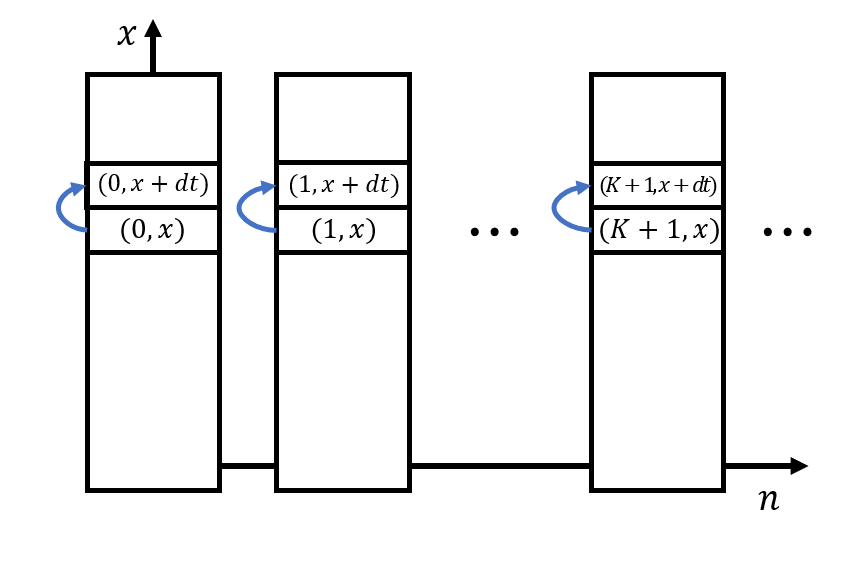}
        \caption{State variables increase in $x$ with no block generation nor transaction arrival during $(x,x+\rd t]$}
        \label{fig:3.4.MG_K1_2}
    \end{subfigure}
    \hfill
    \begin{subfigure}[htb!]{1\textwidth}
        \centering
        \includegraphics[width=0.5\textwidth]{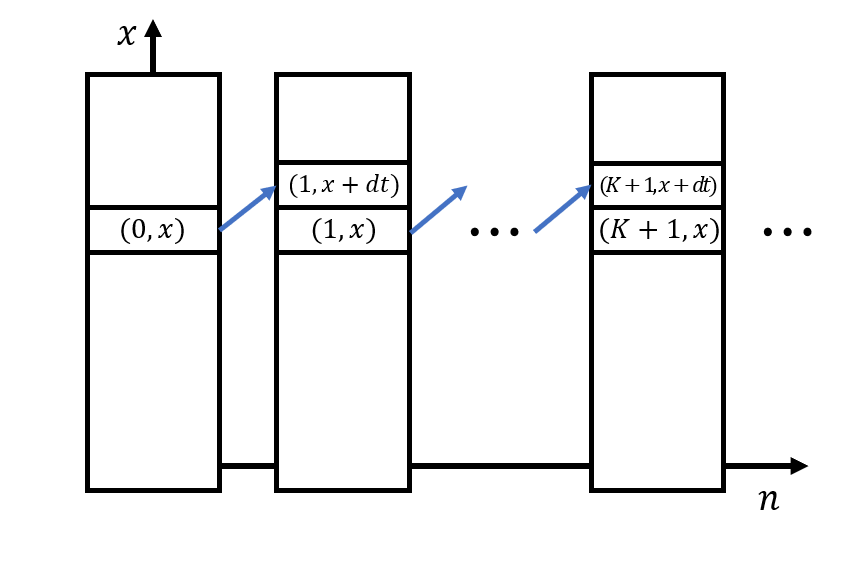}
        \caption{State variables increase simultaneously with transaction arrival during $(x, x+\rd t]$}
        \label{fig:3.4.MG_K1_3}
    \end{subfigure}
    \hfill
    \begin{subfigure}[htb!]{1\textwidth}
        \centering
        \includegraphics[width=0.5\textwidth]{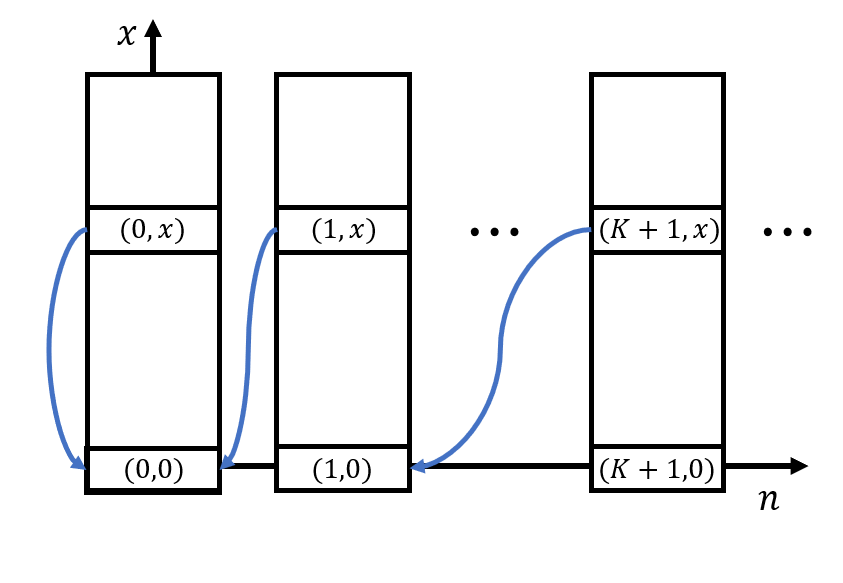}
        \caption{State variables become $(0,0)$ or $(n, 0)$ with block creation during $(x, x+\rd t]$}
        \label{fig:3.4.MG_K1_4}
    \end{subfigure}
       \caption{State transitions in the transaction queue where states are (number of transactions, elapsed time).}
   \label{fig:queue_dynamics}
\end{figure}

Let us first consider the embedded chain $\{N_k\}$ defined by $N_k = N(T_k -)$. Here, $T_k$ is the generation time of the $k$-th block. Since arrivals are Markovian, $\{N_k\}$ is a discrete-time Markov chain. It is a simple matter to check this process satisfies 
$$
N_k = (N_{k-1} - K)^+ + A_k
$$
where $A_k$ is the number of transaction arrivals between $T_{k-1}$ and $T_k$. Clearly, $A_k$ is a Poisson random variable with mean $\lambda \mu $ and it is independent of $N_{k-1}$. The positive recurrence of $\{N_k\}$ can be shown by the Foster-Lyapunov criterion. With a Lyapunov function specified by $v(n)=n$, the so called expected drift of $v$ is 
$$
\Delta v(n)  = \mathbb{E}\big[ v(N_k) - v(N_{k-1})| N_{k-1}=n\big] = \lambda \mu - \min\{K, n\}. 
$$
Due to our assumption that $\rho = \lambda \mu / K < 1$, we see $\Delta v(n)$ is negative for any $n \geq K$. This negative drift outside a finite set ensures the positive recurrence of $\{N_k\}$. It is standard to argue that the process $(N(t), X(t))$ is also positive recurrent thanks to the fact that $(N(t), X(t))$ regenerates at $T_k$'s in addition to the positive recurrence of $\{N_k\}$. Following \cite{downton1956limiting, chaudhry1983first}, we then derive differential equations that characterize stationary probabilities. 

For this, let us consider the density function $\pi_n(x)$ as the limit of $\pi_n(x,t)$ which is defined by $\pi_n(x,t)  \rd x = \bbP\lt( N(t)= n, x < X(t) \leq x + \rd x \rt)$. Here $\pi_n(x,t)\rd x$ is the probability of having $n$ in the system and the elapsed time since last block generation being in $(x, x+\rd x]$.  
From the Poisson arrival, we have the probaiblity $\lambda \rd x$ of new arrival as time elapses. We also have the probability $\eta(x) \rd x$ of new block generation where $\eta(x)$ is the conditional service rate or the hazard rate $f_B(x) / \bar F_B(x)$. Therefore, the  likelihood $\pi_n(x+\rd x, t + \rd x)$ satisfies 
$$
\pi_n(x+\rd x,t+\rd x) = (1-\lambda \rd x)\pi_n(x,t) (1-\eta(x)\rd x)+\lambda \rd x\cdot \pi_{n-1}(x,t) (1-\eta(x)\rd x)+o(\rd x) 
$$
for $n=1,2, \ldots$. Instead of solving this equation for $(n,x, t)$, we are interested in limiting or equilibrium quantities as typically done in the analysis of queueing systems. Rearranging terms and given the existence of limiting distributions, we get 
\begin{equation}\label{eq-ode_n}
\frac{\rd \pi_n(x)}{\rd x} = - (\lambda + \eta(x))\pi_n(x) + \lambda \pi_{n-1}(x).
\end{equation}
For $n=0$, we note that $\pi_0(x+\rd x, t + \rd x) = (1 - \lambda \rd x) \pi_0(x,t) (1 - \eta(x)\rd x) + o(\rd x)$ and this leads to 
\begin{equation}\label{eq-ode_0}
\frac{\rd \pi_0(x)}{\rd x} = - (\lambda + \eta(x))\pi_0(x). 
\end{equation}

In addition to these differential equations, there are also boundary conditions for $\pi_0(0)$ and $\pi_n(0)$ that are the limiting probabilities of $\pi_0(0, t)$ and $\pi_n(0,t)$ right after new block generation. For the former, the event occurs whenever a new block is generated for the queue size less than or equal to $K$. Note that there is a possibility of an empty block. Hence, we have
\begin{equation}\label{eq-bdry_0}
\pi_0(0) = \sum_{n=0}^K \int_0^\infty \pi_n(x) \eta(x) \rd x.
\end{equation}
For the latter, the event of having $n$ transactions in the queue at the time of block generation is relevant only when there are $n+K$ transactions in the pre-completion queue. Therefore, 
\begin{equation}\label{eq-bdry_n}
\pi_n(0) = \int_0^\infty \pi_{n+K}(x) \eta(x) \rd x
\end{equation}
for $n=1, 2, \ldots$. Lastly, the probability distribution $\{\pi_n(x)\}$ satisfies $\sum_{n=0}^\infty \int_0^\infty \pi_n(x) \rd x = 1$. 
Using the above equations, we shall derive expressions for key performance metrics such as expected waiting time and expected queue length. For this, we calculate the probability generating function $\Pi(z; x)$ for the distribution $\{\pi_n(x)\}$ at elapsed time $x$. It is a variant of results found in  Chapter 4 of \cite{chaudhry1983first}. We also notice that similar derivations are done in \cite{kawase2017transaction} and \cite{kasahara2019effect}. The difference is that we have $\pi_0(x)$ here because blocks are produced regardless of the presence of transaction requests.

\begin{lem}\label{lem-pgf}
    The probability generating function $\displaystyle\Pi(z;x)=\sum_{n=0}^\infty \pi_n(x)z^n$ for the queue length at elapsed time $x$ is given by 
\begin{equation*}
\Pi(z; x) = \Pi(z; 0) \bar F_B(x) e^{-\lambda(1-z)x}.
\end{equation*}
Here, $\Pi(z;0) = \frac{\sum_{n=0}^K(z^K-z^n)\int_0^\infty \pi_n(x) \eta(x)\rd x}{z^K-\beta(\lambda(1-z))}$ with $\Pi(1; 0) = \mu^{-1}$ and $\beta(t)$ is the Laplace transform $\displaystyle\beta(t)=\int_0^\infty e^{-tx}f_B(x)\rd x$ of the block generation time. 
\end{lem}
\begin{proof}
When we solve \eqref{eq-ode_0}, we easily get $\pi_0(x) = \pi_0(0) \bar F_B(x) e^{-\lambda x}$. It is also easy to see that \eqref{eq-ode_n}  yields 
\begin{equation}\label{eq-lem_temp}
\pi_n(x) = e^{-\lambda x}\Bar{F}_B(x)\left(\pi_n(0)+\lambda\int_0^x \frac{e^{\lambda s}}{\Bar{F}_B(s)}\pi_{n-1}(s)\rd s\right).
\end{equation}
Then, we proceed as follows: 
\begin{eqnarray*}
        \Pi(z;x) &=&\sum_{n=0}^\infty \pi_n(x)z^n\\
        & =& e^{-\lambda x}\Bar{F}_B(x)\left(\pi_0(0)+\sum_{n=1}^\infty\left(\pi_n(0)+\lambda \int_0^x \frac{e^{\lambda s}}{\Bar{F}_B(s)}\pi_{n-1}(s)\rd s\right)z^n\right)\\
        &=& e^{-\lambda x}\Bar{F}_B(x)\left(\Pi(z;0)+\lambda z\sum_{n=0}^\infty z^n\int_0^x\frac{e^{\lambda s}}{\Bar{F}_B(s)}\pi_{n}(s)\rd s\right).
\end{eqnarray*}
We used \eqref{eq-lem_temp} in the second equality. In the third line of the equation, we replace $\pi_n(s)$ with the right hand side of \eqref{eq-lem_temp} for $n \geq 1$. For $n=0$, the term is simply $\pi_0(0)x$. Then, we get 
\begin{eqnarray*}
 \Pi(z;x)        & =& e^{-\lambda x}\Bar{F}_B(x)\left(\Pi(z;0)+ \lambda z x \pi_0(0) + \lambda z\sum_{n=1}^\infty z^n\int_0^x
        \left[ \pi_n(0) + \lambda \int_0^s \frac{e^{\lambda s_1}}{\Bar{F}_B(s_1)} \pi_{n-1}(s_1)\rd s_1 \right] \rd s\right)\\
        &= &e^{-\lambda x}\Bar{F}_B(x)\left(\Pi(z;0)(1+\lambda z x)+(\lambda z)^2\sum_{n=0}^\infty z^n\int_0^x \int_0^{s} \frac{e^{\lambda s_1}}{\Bar{F}_B(s_1)}\pi_{n}(s_1)\rd s_1 \rd s\right)\\
        &= &e^{-\lambda x}\Bar{F}_B(x)\left(\Pi(z;0)\sum_{n=0}^m\frac{(\lambda z x)^n}{n!}+(\lambda z)^{m+1}\sum_{n=0}^\infty z^n\int_0^x \cdots \int_0^{s_{m-1}} \frac{e^{\lambda s_m}}{\Bar{F}_B(s_m)}\pi_{n}(s_m)\rd s_m \cdots \rd s\right)\\
        &=& \Pi(z;0)\Bar{F}_B(x) e^{-\lambda(1-z)x}.
\end{eqnarray*}
Here, the third equality is obtained by repeating the same argument. The last equality is then obtained by sending $m$ to infinity. The convergence of the first term is trivial. The second term can be shown to converge to zero by obtaining an upper bound as 
\begin{eqnarray*}
\lefteqn{ 
|\lambda z|^{m+1}\sum_{n=0}^\infty   |z|^n    \int_0^x \cdots \int_0^{s_{m-1}} \frac{e^{\lambda s_m}}{\Bar{F}_B(s_m)}\pi_{n}(s_m)\rd s_m \cdots \rd s} && \\
& \leq & |\lambda z|^{m+1} \sum_{n=0}^\infty |z|^n \max_{s\in[0,x]}\frac{e^{\lambda s}}{\Bar{F}_B(s)} 
\int_0^x \cdots \int_0^{s_{m-1}} \pi_n(s_m) \rd s_m \cdots \rd s \\
& \leq & |\lambda z|^{m+1} \sum_{n=0}^\infty |z|^n \max_{s\in[0,x]}\frac{e^{\lambda s}}{\Bar{F}_B(s)} 
\int_0^x \cdots \int_0^{s_{m-2}} \int_0^\infty  \pi_n(s_m) \rd s_m \rd s_{m-1} \cdots \rd s \\
& \leq & |\lambda z|^{m+1} \sum_{n=0}^\infty |z|^n \max_{s\in[0,x]}\frac{e^{\lambda s}}{\Bar{F}_B(s)} 
\int_0^x \cdots \int_0^{s_{m-2}} \pi_n  \rd s_{m-1} \cdots \rd s \\
& \leq & |\lambda z|^{m+1} \sum_{n=0}^\infty |z|^n \max_{s\in[0,x]}\frac{e^{\lambda s}}{\Bar{F}_B(s)} 
\frac{x^m \pi_n}{m!} = |\lambda z|^{m+1} \Pi(|z|)  \max_{s\in[0,x]}\frac{e^{\lambda s}}{\Bar{F}_B(s)} 
\frac{x^m}{m!}.
\end{eqnarray*}
Here, $\pi_n = \int_0^\infty \pi_n(x) \rd x$ and $\Pi(z)$ is the probability generating function of $\{\pi_n\}$ which will be formally introduced in the next subsection. It is clear that the right hand side converges to zero as $m$ increases. As a consequence, the first claim in the statement of the lemma is proved. 

For the computation of $\Pi(z;0)$, we integrate $\Pi(z;x)$ with respect to the hazard rate
\begin{eqnarray*}
\int_0^\infty \Pi(z; x) \eta(x) \rd x &=& \Pi(z;0) \int_0^\infty \bar F_B(x) e^{-\lambda (1-z)x} \eta(x) \rd x \\
&=& \Pi(z;0) \int_0^\infty  e^{-\lambda (1-z)x} f_B(x) \rd x \\
&=& \Pi(z;0) \beta(\lambda(1-z)).
\end{eqnarray*}
On the other hand, the left hand side can also be written as 
\begin{eqnarray*}
\int_0^\infty \Pi(z;x) \eta(x) \rd x &=& \int_0^\infty \sum_{n=0}^\infty \pi_n(x)z^n \eta (x) \rd x \\
&=& \sum_{n=0}^K z^n \int_0^\infty \pi_n(x) \eta(x) \rd x + \sum_{n=K+1}^\infty z^n \int_0^\infty \pi_n(x) \eta(x) \rd x \\
&=& \sum_{n=0}^K z^n \int_0^\infty \pi_n(x) \eta(x) \rd x + z^K \sum_{n=1}^\infty z^n \int_0^\infty \pi_{n+K}(x) \eta(x) \rd x \\
&=& \sum_{n=0}^K z^n \int_0^\infty \pi_n(x) \eta(x) \rd x + z^K \sum_{n=1}^\infty z^n \pi_n(0) \\
&=& \sum_{n=0}^K z^n \int_0^\infty \pi_n(x) \eta(x) \rd x + z^K \big( \Pi(z;0) - \pi_0(0) \big)\\
&=& \sum_{n=0}^K \left(z^n - z^K \right) \int_0^\infty \pi_n(x) \eta(x) \rd x + z^K  \Pi(z;0).
\end{eqnarray*}
Here, \eqref{eq-bdry_n} is used in the fourth equality and \eqref{eq-bdry_0} in the last equality. By equating the two expressions above, we obtain 
$$
\Pi(z;0) = \frac{ \sum_{n=0}^K \left(z^K - z^n \right) \int_0^\infty \pi_n(x) \eta(x) \rd x}{ z^K - \beta(\lambda(1-z))}.
$$

It remains to show $\Pi(1;0) = \mu^{-1}$. The probability distribution $\{\pi_n\}$ satisfies $\sum_{n=0}^\infty \pi_n = 1$. Therefore, we have
\begin{eqnarray*}
1 &=& \lim_{z \rightarrow 1} \sum_{n=0}^\infty \int_0^\infty \pi_n(x) z^n \rd x \\
&=& \lim_{z \rightarrow 1} \int_0^\infty \Pi(z; x) \rd x \\
&=& \Pi(1;0) \lim_{z \rightarrow 1} \frac{1 - \beta(\lambda(1-z))}{\lambda(1-z)} \\
&=& \Pi(1;0) \mu
\end{eqnarray*}
where the third equality uses the integration by parts and the last equality is obtained by L'H\^{o}pital's rule.
\qed
\end{proof}

\subsection{Expected Queue Length}

The queue dynamics may be summarized by the queue length distribution and its mean. By integrating the limiting distribution with respect to the elapsed time, we get the probability of having $n$ transactions in the mempool at equilibrium, i.e. $\pi_n := \int_0^\infty \pi_n(x) \rd x$ for $n=0, 1, \ldots$. The associated probability generation function is defined by $\Pi(z) = \sum_{n=0}^\infty \pi_n z^n$. Then, observe that 
\begin{eqnarray*}
\Pi(z) &=& \sum_{n=0}^\infty \int_0^\infty \pi_n(x) z^n \rd x \\
&=& \int_0^\infty \Pi(z; x) \rd x \\
&=& \Pi(z; 0) \int_0^\infty \bar F_B(x) e^{-\lambda(1-z)x} \rd x\\
&=& \Pi(z; 0) \frac{1 - \beta(\lambda(1-z))}{\lambda(1-z)}. 
\end{eqnarray*}
Here, it is implicitly assumed that the integral on the right hand side of the third equality is finite. One sufficient condition, for example, is $\liminf_{x \rightarrow \infty} \eta(x)  > 0$ so that the last equality is valid for $z < 1  + \varepsilon$ for some sufficiently small positive $\varepsilon$. 

The expected queue length is then computed by $\Pi'(1)$. Since the mean block generation time $\mu$ and the maximum capacity $K$ are assumed to be fixed in the blockchain network, we view the expected queue length as a function of the load parameter $\rho$ and write 
\begin{eqnarray*}
Q(\rho)  &=& \Pi'(1) \\
&=& \Pi'(1; 0) \mu + \Pi(1;0) \frac{\rd }{\rd z} \Big|_{z=1} \frac{1 - \beta(\lambda(1-z))}{\lambda(1-z)} \\
&=& \Pi'(1; 0) \mu+ \Pi(1;0)\frac{\lambda}{2} \mathbb{E}\left[B^2\right]  \\
&=& \Pi'(1;0) \mu + \frac{\rho K}{2}\left(1 + \frac{\sigma^2}{\mu^2} \right)
\end{eqnarray*}
where $\sigma^2$ is the variance of the block generation time $B$ and $\Pi(1;0)$ is given in Lemma~\ref{lem-pgf}. The lemma also shows
\begin{equation}\label{eq-pi_temp}
\Pi(z;0) D(z) = \sum_{n=0}^K \left(z^K - z^n\right)  \int_0^\infty \pi_n(x)\eta(x) \rd x
\end{equation}
where $D(z) = z^K - \beta(\lambda(1-z))$ for notational simplicity. By directly differentiating the both sides twice, then one readily obtains a closed form expression for $\Pi'(1;0)$. Nevertheless, it involves the computation of $\int_0^\infty \pi_n(x) \eta(x) \rd x$ for $n=0, 1, \ldots, K-1$. It dates back to 50's when those values are found by searching for roots of $D(z)$ on and within the unit circle on the complex plane. See \cite{bailey1954queueing}. However, as pointed out in \cite{oblakova2019exact}, it is not trivial to find such zeros without any closed-form expression. Approximate methods may pose a problem as the precision of such values has a high impact. In our case, the scale itself is a challenge as $K$ is typically several thousands. Instead, we adopt the proposed method of \cite{oblakova2019exact} who apply their idea to discrete-time queueing systems with traffic applications. 

\begin{prop}\label{prop-queue_length}
     There exists a sufficiently small $\varepsilon$ such that $D(z)$ has no zero in the annulus $\{z \in \mathbb{C} : 1 < |z| \leq 1 + \varepsilon\}$ whereas there are $K-1$ solutions in $\{ z \in \mathbb{C}: |z| < 1\}$. With such $\varepsilon$, the expected queue length is then given by
    \begin{equation*}
        Q(\rho)=\frac{1}{2\pi }\int_{-\pi}^{\pi}\frac{D'(z(\varphi))}{D(z(\varphi))}\frac{z(\varphi)}{1-z(\varphi)}\rd \varphi + \frac{\rho K}{2}\left(1+\frac{\sigma^2}{\mu^2}\right)
    \end{equation*}
where $z(\varphi) = (1+\varepsilon)e^{i \varphi}$. 
\end{prop}
\begin{proof}
For any $z \in \mathbb{C}$ on the circle of radius $r > 1$, say $S_r$, we observe 
$$
    \left |\beta(\lambda(1-z))\right|\leq \int_0^\infty |e^{-\lambda(1-z)x}f_B(x)|\rd x = \int_0^\infty e^{- \lambda x(1-r \cos \varphi)} f_B(x)\rd x
$$
where $z = r e^{i \varphi}$. This is in turn less than or equal to $\int_0^\infty e^{-\lambda x (1-r)} f_B(x) \rd x = \beta(\lambda(1-r))$. Using $|z^K| = r^K$ on $S_r$, we have the relationship $D(r) =  r^K - \beta(\lambda(1-r))$.
Note that $D(1) = 0$ and $D'(1) = K - \lambda \mu = (1 - \rho)K$. Since we assume that the system is not eplosive or $\rho < 1$, we have $D'(1) > 0$ and thus $D(r)  > 0$ for any $r$ sufficiently close to 1. For such $r$, we have $|\beta(\lambda(1-z))| < |z^K|$ on $S_r$. Then, by Rouch\'e's Theorem (one version of which we cite below to make the presentation self-contained), we conclude $z^K$ and $D(z)$ have the same number of zeros, i.e. $K$, inside $S_r$, counting multiplicities.

\begin{lem}[Rouch$\acute{\text{e}}$'s Theorem]
    Suppose two complex valued functions $f$ and $g$  are analytic inside some region $U$ with a  simple closed contour $\partial U$. If $|g(z)| < |f(z)|$ on $\partial U$, then $f$ and $f + g$ have the same number of zeros inside $U$, counting multiplicities. 
\end{lem}

Due to the analyticity of $D(z)$, its zeros are isolated. In particular, there are finitely many zeros on the compact set $\{ z \in \mathbb{C}: 1 \leq |z| \leq r\}$ for any $r > 1$. Thus by selecting a sufficiently small $\varepsilon > 0$, we can assure that any zero of $D(z)$ on $\{z \in \mathbb{C} : 1 \leq |z| \leq 1 + \varepsilon\}$ has the radius 1. Clearly, $D(1) = 0$. For $z \neq 1$ on the unit circle, we note 
$$
|\beta(\lambda(1-z))| \leq \int_0^\infty e^{-\lambda x(1- \cos \varphi)} f_B(x) \rd x
$$
is strictly less than $|z^K| = 1$ with $z = e^{i \varphi}$. Hence, $z=1$ is the only solution of $D(z) = 0$ on the unit circle. Further, we already noted above that $D'(1) > 0$. This implies that the multiplicity of $z = 1$ is 1. Consequently, there are $K-1$ solutions to $D(z) = 0$ within the unit circle and a simple unique solution $z=1$ on the unit circle.

Now we are ready to apply the solution procedure proposed in \cite{oblakova2019exact}. To be specific, let $\{z_i\}_{i=1}^{K-1}$ be the $K-1$ roots of $D(z)=0$ in the interior of the unit ball. We also define $a_k =\sum_{n=0}^k\int_0^\infty \pi_n(x)\eta(x)\rd x$, the partial sum of the unknowns for $k=0, 1, \ldots, K-1$. Then, it is easy to see that the right hand side of \eqref{eq-pi_temp} is equal to $(z-1) \sum_{n=0}^{K-1} a_n z^n$. Since $D(z_i) = 0$ for $i=1, \ldots, K-1$, it must be that 
$$
\sum_{n=0}^{K-1} a_n z^n = a_{K-1} \prod_{i=0}^{K-1} (z - z_i).
$$
This expression further yields 
$$
\Pi(1;0) =\lim_{z \rightarrow 1} \frac{z-1}{D(z)} a_{K-1} \prod_{i=0}^{K-1} (1-z_i) = \frac{a_{K-1}}{(1-\rho)K}  \prod_{i=0}^{K-1} (1-z_i).
$$
Since $\Pi(1;0) = \mu^{-1}$ from Lemma~\ref{lem-pgf}, we get $a_{K-1} = \mu^{-1} (1-\rho)K   \prod_{i=0}^{K-1} (1-z_i)^{-1}$ and finally, 
$$
\Pi(z; 0) = \frac{(1-\rho)K (z-1)}{\mu \; D(z)} \prod_{i=1}^{K-1} \frac{z-z_i}{1-z_i}.
$$

The computation of the expected queue length $Q(\rho)$ requires $\Pi'(1;0)$. For this, define $h(z) = \prod_{i=1}^{K-1} \frac{z-z_i}{1-z_i}$ for convenience, and observe 
\begin{eqnarray*}
        \Pi'(1;0) &=& \frac{(1-\rho)K}{\mu}  \lim_{z\to 1}\frac{D(z)-(z-1)D'(z)}{D(z)^2}\cdot h(1)+\frac{1}{\mu} h'(1)\\
        &=& \frac{1}{\mu}\left[ (1-\rho)K \lim_{z\to 1}\frac{D(z)-(z-1)D'(z)}{D(z)^2}+\sum_{i=1}^{K-1}\frac{1}{1-z_i}\right]\\
        &=& \frac{1}{\mu}\left[- \frac{1}{2} \lim_{z\to 1}\frac{D''(z)}{D'(z)}+\sum_{i=1}^{K-1}\frac{1}{1-z_i}\right]\\
        &=& \frac{1}{\mu}\left[ \frac{1}{2(1-\rho)} \left( K \rho^2 \left( 1 + \frac{\sigma^2}{\mu^2}\right) - (K-1) \right) +\sum_{i=1}^{K-1}\frac{1}{1-z_i}\right].
\end{eqnarray*}
We utilize a lemma of \cite{oblakova2019exact}, which we state below for the reader's convenience. 

\begin{lem}[\cite{oblakova2019exact}]\label{lem-oblakova}
    Consider a function $F(z)$ that is analytic in a neighborhood of $z_i$ , where $i = 0,\cdots, K - 1$. Then, the residue of the function $\displaystyle\frac{D'(z)}{D(z)}F(z)$ at $z_i$ is equal to
    \begin{equation*}
        \lim_{z\to z_i}\frac{D'(z)}{D(z)}F(z)(z-z_i)=m_{z_i}F(z_i)
    \end{equation*}
where $m_{z_i}$ is the multiplicity of the root $z_i$.
\end{lem}
Treating $S_{1 + \varepsilon}$ as a closed curve with counterclockwise orientation, the Residue Theorem says that the contour integral $\oint_{S_{1+\varepsilon}} \frac{D'(z)}{D(z)} \frac{1}{1-z}\rd z$ is equal to $2\pi i$ times the sum of residues at $z_i$'s for $i=1, \ldots, K-1$ as well as the residue at $z=1$. Lemma~\ref{lem-oblakova} then implies 
$$
\frac{1}{2\pi i}\oint_{S_{1+\varepsilon}} \frac{D'(z)}{D(z)} \frac{1}{1-z}\rd z = (\text{residue at 1})  + h'(1).
$$
The residue at 1 is found by repeatedly applying L'H\^opital's rule and given as follows:
\begin{eqnarray*}
\text{residue at 1} &=& \lim_{z \rightarrow 1} \frac{\rd}{\rd z} \left[ \frac{D'(z)}{D(z)} \frac{1}{1-z} (z-1)^2 \right] \\
&=& - \frac{D''(1)}{2D'(1)} \\
&=&  \frac{1}{2(1-\rho)}\left( K \rho^2 \left( 1+ \frac{\sigma^2}{\mu^2} \right) - (K-1) \right).
\end{eqnarray*}
Consequently, we obtain a succinct expression 
$$
\Pi'(1;0) = \frac{1}{2\pi i \mu} \oint_{S_{1+\varepsilon}} \frac{D'(z)}{D(z)} \frac{1}{1-z} \rd z.
$$
By re-writing $z = (1+\varepsilon) e^{i \varphi}$ for $ -\pi < \varphi \leq \pi$, we get the desired result. 
\qed
\end{proof}
The above proposition provides us with a numerical tool to evaluate the expected queue length when the block generation time is different from the benchmark Bitcon network. More specifically, we infer user behaviors or cost structure based on the Bitcoin network, and then utilize the result for optimal bidding strategies in other blockchain networks. 

\section{Strategic Bidding}\label{sec:bidding}

The previous section on expected queue length does not distinguish transaction requests in the mempool. However, miners in the network naturally select transactions with highest possible fees to maximize rewards given the constraint on the block size. In this sense, users compete with each other to make their transactions complete. Each user must determine the level of bidding by gauging the delay cost and the level of congestion. For this purpose, we need a model for the user's optimal decision. The main reference we follow is \cite{huberman2021monopoly} in which the net reward for the $i$-th user with delay cost $c_i$ is given by $R_i - b_i - c_i W_i$ where $R_i$ is the willingness-to-pay, $b_i$ is the bid, and $W_i$ is the expected waiting time of the user. Here $c_i$ is the cost per unit time until the user's transaction goes through. Also, $R_i$ is the benefit for the user in using the blockchain network over other alternatives that fulfill the same objective of the transaction. In \cite{huberman2021monopoly}, $R_i$ is assumed to be either low or high value and it is not correlated with the cost $c_i$. As such, this is non-essential for our purpose of studying the user behavior through bidding data, hence we only consider the utility of the form $-b_i - c_i W_i$.

In this queueing game where higher priorities are earned by higher bids, the expected waiting time $W_i$ is given by a function of the bid, say $\bW(b_i)$. When the user decides the bid optimally, we can even further write $b_i$ as a function of $c_i$, say $b_i = \bb(c_i)$ so that the net reward for the $i$-th user can be re-written as $- \bb(c_i) - c_i \bW(\bb(c_i))$. In addition, users are heterogeneous so that their delay costs $c$'s are assumed to have a cumulative distribution $F$ with density $f$. Since it is natural to think of the optimal bidding at equilibrium as a strictly increasing and continuous function of $c$, this leads to the distribution of bids $b$'s, say $G$ with density $g$. This means that $\mathbb{P}(\mathbcal{c}  > c_i) = \bar F(c_i)$ whereas $\mathbb{P}(\mathbcal{c}  > c_i) = \mathbb{P}(\mathbcal{b} > \bb(c_i)) = \bar G(\bb(c_i))$. Here $\mathbcal{c}, \mathbcal{b}$ are the random delay cost and equilibrium bid of the user population. When the $i$-th user arrives at the mempool with the equilibrium bid, the user sees the arrival rate of higher bids equal to $\lambda \bar G(\bb(c_i)) = \lambda \bar F(c_i)$. 

For the rest of the paper, we assume that $F$ and $G$ are continuous and strictly increasing for analytical tractability. We in addition make a notational change. Since each user experiences different levels of congestion due to priority, we use $\bar\rho = \lambda\mu/K$ the maximum possible load for the whole queue and leave $\rho$ for the generic load parameter. If we define a function $\brho(c_i)$ as the load parameter for the blockchain queue with arrival rate $\lambda \bar F(c_i)$, then $\brho(c_i) = \bar\rho \bar F(c_i)$ with state space $[0, \bar\rho]$. Since this function is continuous and invertible, we can assign the $i$-th user the corresponding load parameter, say $\rho_i = \brho(c_i)$. In other words, the priority of the $i$-th user in the original priority queue is represented by the load parameter $\rho_i$ where a lower value means a higher priority. We also note that any randomly selected user has a load parameter $\brho(\mathbcal{c})$ for the random delay cost. This implies that $\mathbb{P}(\brho(\mathbcal{c}) \leq \rho_i) = \mathbb{P}(\brho(\mathbcal{c}) \leq \brho(c_i)) = \mathbb{P}(\mathbcal{c} \geq c_i) = \bar F(c_i) = \rho_i / \bar\rho$. In other words, $\brho(\mathbcal{c})$ has a uniform distribution on $[0, \bar\rho]$.  

For our M/G\textsuperscript{$K$}/1 queuing of transactions but with the load parameter $\rho_i$, the expected waiting time for any user with the load parameter between 0 and $\rho_i$ is given as $Q(\rho_i)/  (\lambda \bar F(c_i))$ by Little's Law. Abusing notation, let us use $\bW(\rho)$ for the expected waiting time that any user with paritcular priority $\rho$ experiences. Since priorities are uniformly distributed and $\bar F(c_i) = \rho_i / \bar \rho$, we obtain 
$$
\frac{Q(\rho_i)}{\lambda \bar F(c_i)} = \frac{\bar\rho Q(\rho_i)}{\lambda \rho_i } = \frac{1}{\rho_i} \int_0^{\rho_i} \bW(\rho) \rd \rho.
$$
This implies that 
\begin{eqnarray*}
\bW(\rho_i) &=& \frac{\rd}{\rd \rho_i} \left[ \frac{\bar \rho}{\lambda} Q(\rho_i) \right] \\
&=& \frac{\mu}{K} Q'(\rho_i) \\
&=& \frac{\mu}{2}\left(1 + \frac{\sigma^2}{\mu^2} \right) + \frac{\mu}{2\pi K} \int_{-\pi}^\pi \bI(\rho_i, \varphi) \rd \varphi
\end{eqnarray*}
where $\bI(\rho_i, \varphi) = \left[ \frac{\rd}{\rd \rho_i} \frac{D'(z(\varphi))}{D(z(\varphi))} \right] \frac{z(\varphi)}{1 - z(\varphi)}$ for notational simplicity. In our numeical implementation, this function is calculated for a given block generation time distribution and enters into the above formula for numerical integration. 

Our ultimate goal is to understand uers' optimal bidding strategy $\bb$. For the reader's convenience, we summarize the procedure to arrive at the final expression of the function. First of all, recall that the user with delay cost $c_i$ has the utility $-b_i - c_i \bW(b_i)$ where $b_i$ is the bid of the $i$-th user, not necessarily optimal. Then, the first order condition to maximize the utility gives us the relationship 
\begin{equation}\label{eq:cost-bid}
\frac{\rd}{\rd b}\Big|_{b = \bb(c_i)} \bW = - \frac{1}{c_i}.
\end{equation}
On the other hand, the relationship between $\bb, \brho, F$ and $G$ implies that 
\begin{eqnarray*}
\frac{\rd}{\rd b}\Big|_{b = \bb(c_i)} \bW &=& \frac{\rd}{\rd \rho}\Big|_{\rho = \brho(c_i)} \bW \cdot \frac{\rd \rho}{\rd b}\Big|_{b = \bb(c_i)} \\
&=& - \frac{\rd}{\rd \rho}\Big|_{\rho = \brho(c_i)} \bW  \cdot  \frac{\bar\rho f(c_i)}{\bb'(c_i)}.
\end{eqnarray*}
Consequently, the desired formula is obtained by 
\begin{equation*}
\bb(c_i) = \bar\rho \int_0^{c_i} c  f(c) \frac{\rd}{\rd \rho}\Big|_{\rho = \brho(c)} \bW  \rd c.
\end{equation*}
Here, we assume $\bb(0) = 0$, i.e., zero bid when there is no impact of transaction delay. Integration by parts then yields 
\begin{eqnarray*}
\bb(c_i) &=& - c_i \bW(\brho(c_i)) + \int_0^{c_i} \bW(\brho(c)) \rd c \\
&= & - \frac{\mu c_i}{2\pi K} \int_{-\pi}^\pi \bI(\brho(c_i), \varphi) \rd \varphi + \frac{\mu}{2\pi K} \int_0^{c_i} \int_{-\pi}^\pi \bI(\brho(c), \varphi) \rd \varphi \rd c.
\end{eqnarray*}
\cite{huberman2021monopoly} derive a closed-form expression for $\bW$, using an exponential distribution for the block generation time. In our case of general distribution, the computation of $\bb$ is only slightly more involved but makes it possible to test various possibilities. For instance, the block generation time may have a very small variation with the target time 12 seconds as in the Ethereum network, or there may be a mining gap due to operating burden on miners as argued in \cite{carlsten2016instability}. Please see Section~\ref{sec:numerical} for details. 

\cite{kim2024mind} study the strategic behaviors of miners in a blockchain network with the proof-of-work consensus protocol. Miners are motivated to solve a complex puzzle for economic rewards consisting of block reward (e.g. 3.125 BTC at the time of this writing) and the accumulated fees for transactions in a new block. The next proposition guages this portion of economic incentives as a function of the elapsed time since the last block generation. We remind the reader that $\Pi(z;x)$ is actually a function of the load parameter or priority level $\rho$. For the sake of simplicity, we define $\xi(\rho) = \Pi'(1;0)$ and $\xi(\rho; x) = \Pi'(1;x)$.

\begin{prop}\label{prop-revenue}
The expected accumulated fee on the event that a block is generated by time $x$ is given by 
$$
\int_0^x m(c^*(y)) \bar F_B(y) \rd y + \frac{\bar \rho K}{\mu^2} \int_0^x \int_{c^*(y)}^\infty \bb(c) f(c) \rd c \times   y \bar F_B(y) \rd y
$$
where $m(u) = \frac{\bar\rho}{2\pi \mu} \int_{u}^\infty \int_{-\pi}^\pi \bI(\brho(c), \varphi) \bb(c) f(c) \rd\varphi \rd c$ and $c^*(x)$ is a solution to $\xi(\brho(c);x) = K f_B(x)$, if exists; otherwise, $c^*(x) = 0$. 
\end{prop}

\begin{proof}
By definition, we have 
$$
\xi(\rho) = \Pi'(1;0) = \frac{1}{2\pi \mu} \int_{-\pi}^\pi \frac{D'(z)}{D(z)} \frac{z}{1-z} \rd \varphi.
$$
Lemma~\ref{lem-pgf} implies that 
\begin{eqnarray*}
    \xi(\rho;x) 
    &=& \Pi'(1;0) \bar F_B(x) + \lambda x \Pi(1;0) \bar F_B(x) \\
    &=& \xi(\rho) \bar F_B(x) + \frac{\rho K}{\mu^2} x \bar F_B(x). 
\end{eqnarray*}
This function represents the expected queue length of all the bids of the users whose delay costs are larger than or equal to $c$ such that $\brho(c) = \rho$, i.e. $\xi(\rho; x) \rd x = \mathbb{E} [ N ; x < X \leq x +\rd x]$. A miner is assumed to collect the best $K$ bids in the mempool, and such bids are transaction requests with delay cost greater than or equal to $c^*$ or priority less than or equal to $\rho^* = \brho(c^*)$ where $\xi(\rho^*; x) = K f_B(x)$. We note that there might be no solution for given $x$ if the maximum load parameter $\bar\rho$ is small. In such a case, a miner would collect all of the existing bids and we set $c^*(x) = 0$. 

Now we observe that the expected number of bids with corresponding delay costs in $(c, c+\rd c)$ when $X \in (x, x+\rd x]$ is given by 
\begin{eqnarray*}
    \xi(\brho(c+\rd c); x) - \xi(\brho(c); x) &=& \xi'(\brho(c); x) \brho'(c) \rd c.
\end{eqnarray*}
Since all such bids are close to $\bb(c)$, the expected fee for such bids is given by $\bb(c)\xi'(\brho(c); x)\brho'(c)\rd c$. Consequently, 
\begin{eqnarray*}
    R(x) &=& \int_{c^*(x)}^\infty \bb(c) \xi'(\brho(c); x) \bar\rho f(c) \rd c \\
    &=& \int_{c^*(x)}^\infty \bb(c) \xi'(\brho(c)) \bar\rho f(c) \rd c \times \bar F_B(x) + \frac{\bar\rho K}{\mu^2} x\bar F_B(x) \cdot \int_{c^*(x)}^\infty \bb(c)  f(c)\rd c.  
\end{eqnarray*}
Note that $m(c^*(x)) = \int_{c^*(x)}^\infty \bb(c) \xi'(\brho(c)) \bar\rho f(c) \rd c$. By integrating this quantity over the interval $[0,x]$, we get the desired result. 

\qed
\end{proof}

The function $R(x)$ in the proof is the expected fees on the event $X \in (x, x+ \rd x]$. The conditional expected fee given $X= x$ is therefore obtained by dividing $R(x)$ by $f_B(x)$. For instance, the amount of initial expected fees from which the total fees start to accumulate is $m(c^*(0))/f_B(0)$. On the other hand, a simple upper bound of $R(x)$ is given by 
$$
R(0) \bar F_B(x) + \frac{\bar \rho K}{\mu^2} x \bar F_B(x) \mathbb{E}[\mathbcal{b}].
$$
Here we use the property that $c^*(x)$ is increasing in $x$. Therefore, the total expected fees until the new block generation is bounded above by 
$$
R(0) \int_0^\infty \bar F_B(x) \rd x + \frac{\bar\rho K}{\mu^2} \mathbb{E}[\mathbcal{b}] \int_0^\infty x \bar F_B(x) \rd x = R(0) \mu + \frac{\bar\rho K}{2}  \left(1 + \frac{\sigma^2}{\mu^2}\right)\mathbb{E}[\mathbcal{b}].
$$
It is worth noting that this is close to the true value in Proposition~\ref{prop-revenue} as $K$ increases to infinity. Hence, we may use the bound as an approximation to the desired quantity as long as the service size $K$ is large enough. For a last comment before we move onto numerical experiments, we see that the total expected fee for a miner for a large $K$ consists of the $R(0)$ part accumulated over the block generation time $B$ and the expected bid $\mathbb{E}[\mathbcal{b}]$ part from the number of bids in the queue, i.e. $Q(\bar\rho)$. 

\section{Numerical Analysis}\label{sec:numerical}

In this section, we apply theoretical results in previous sections to the Bitcoin network. We specifically target the derivation of users' delay cost structure, and see how users would react to changes in block distribution. Furthermore, we test how such delay cost structures differ across chains. In doing so, the equation \eqref{eq:cost-bid} is essential. We provide details in the following subsections. 


\subsection{Cost Distribution of Bitcoin Users}
\begin{figure}[!ht]
    \centering
    \begin{subfigure}[b]{0.48\textwidth}
        \includegraphics[width=\textwidth]{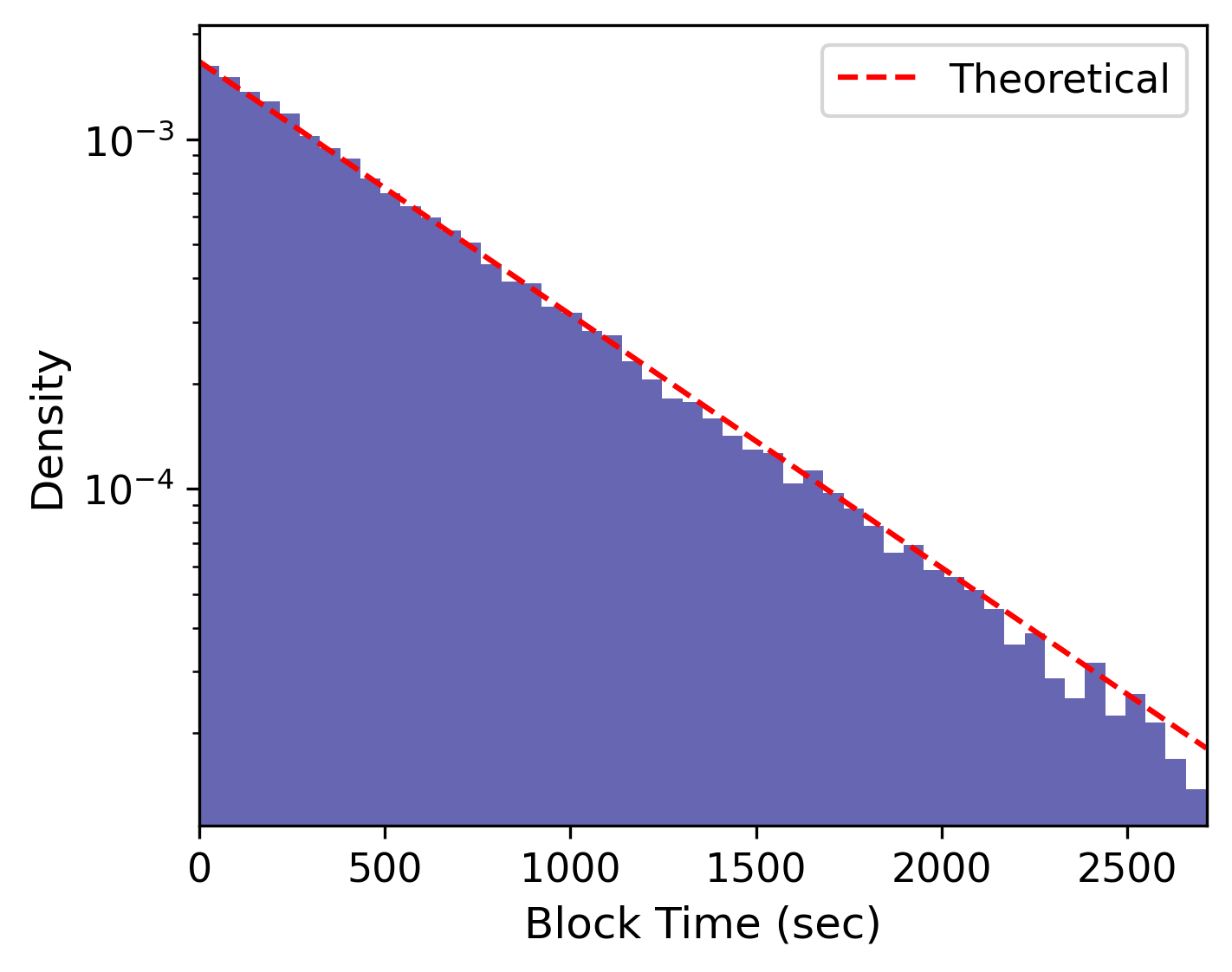}
        \caption{Block time distribution}
        \label{fig:btc_blocktime}
    \end{subfigure}
    \hfill
    \begin{subfigure}[b]{0.48\textwidth}
        \includegraphics[width=\textwidth]{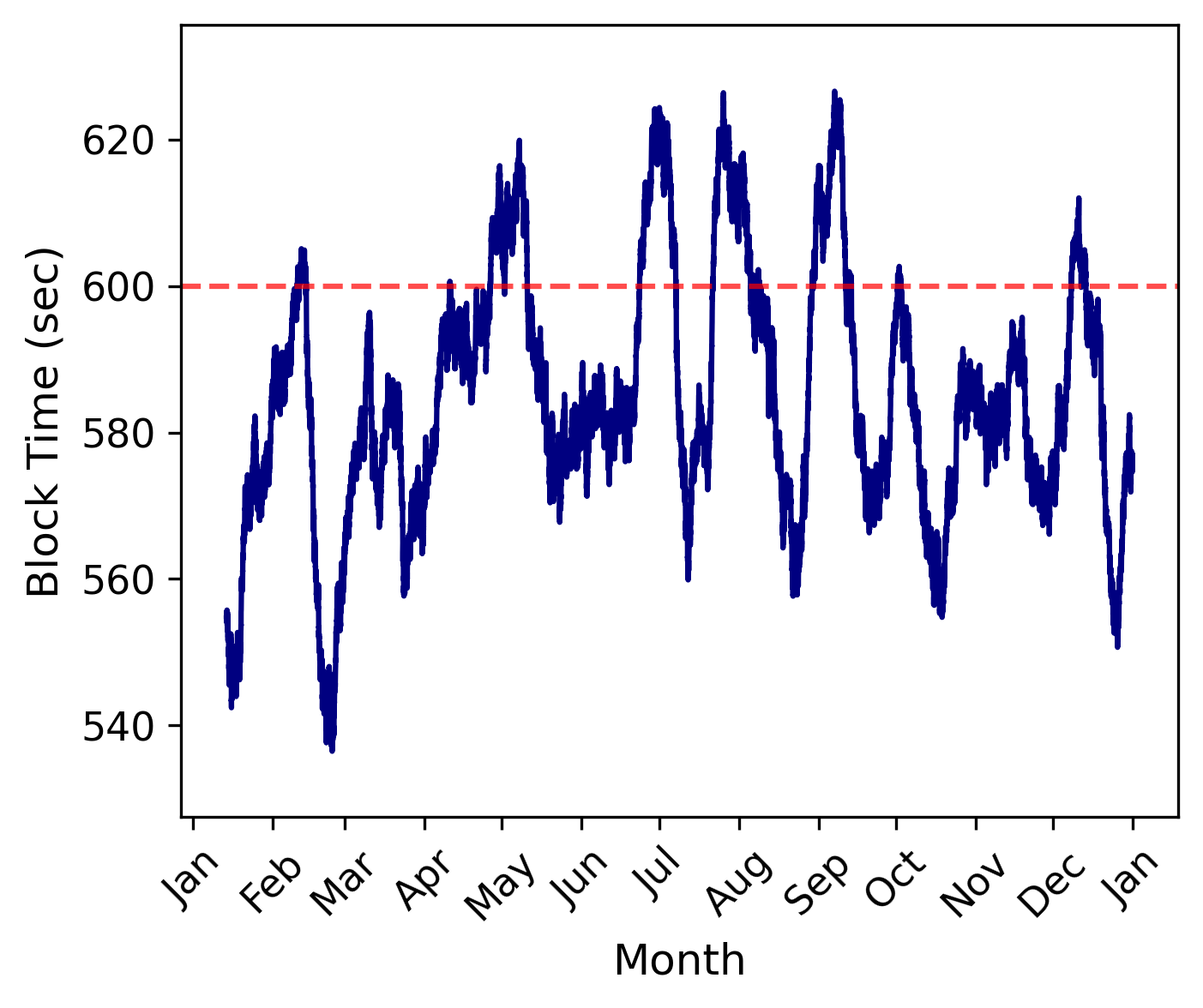}
        \caption{Block time stability}
        \label{fig:btc_blockstability}
    \end{subfigure}
    \hfill
    \begin{subfigure}[b]{0.48\textwidth}
        \includegraphics[width=\textwidth]{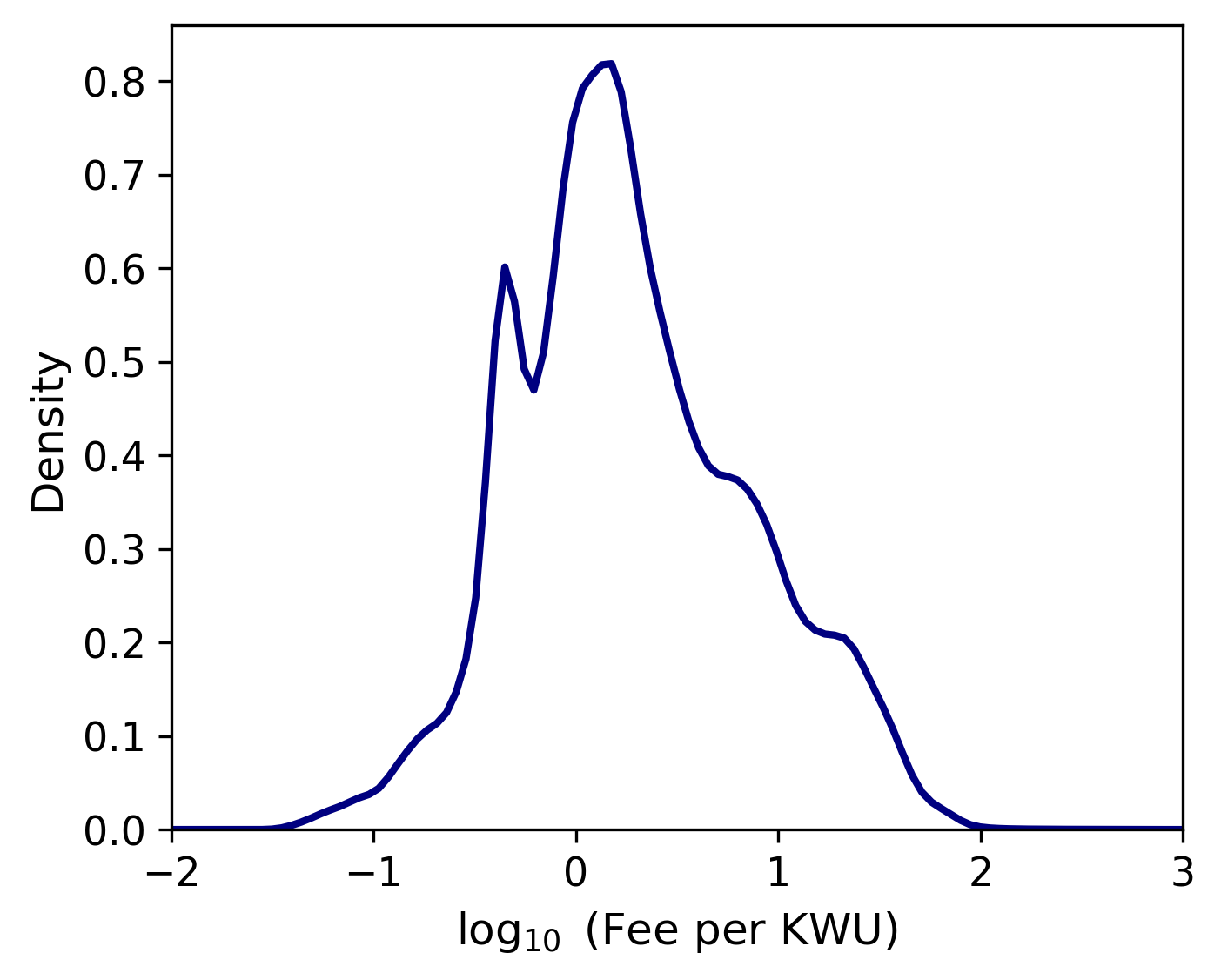}
        \caption{Fee distribution}
        \label{fig:btc_feedist}
    \end{subfigure}
    \caption{Bitcoin network characteristics in 2023}
    \label{fig:btc_analysis}
\end{figure}

We analyze Bitcoin network data from 2023 to understand the empirical characteristics of block generation and transaction processing. The network generated 53,852 blocks during this period, processing approximately 153.4 million transactions. The model parameter $K$ is computed by dividing the block size by the average transaction size, and it is 2,922. For the parameter $\bar\rho$, we divide the average number of transactions per block by $K$, and it is 0.97462. These values represent the actual throughput of the Bitcoin network and form the basis of our waiting cost analysis.

The empirical analysis of the Bitcoin network is shown in Figure~\ref{fig:btc_analysis}. The block time distribution in Figure~\ref{fig:btc_blocktime} plots the empirical block generation intervals, where block time is calculated as $t_i - t_{i-1}$ for timestamp $t_i$ of block $i$. The straight dashed line represents the density of the exponential distribution with mean 600 seconds in log scale. This theoretical model aligns well with the observed block time distribution. Figure~\ref{fig:btc_blockstability} presents the moving average of block generation times using a window size of 2016 blocks, corresponding to Bitcoin's difficulty adjustment period. The average annual block time has remained stable around 10 minutes throughout the year, so we set $\mu=10$ in our model. Figure~\ref{fig:btc_feedist} shows the smoothed probability density of logarithmic transaction fees per 1,000 weight unit (KWU) in USD. It exhibits two distinct peaks and this suggest the presence of user groups with different sensitivities to transaction costs.

\begin{figure}[!ht]
    \centering
    \begin{subfigure}[b]{0.48\textwidth}
        \includegraphics[width=\textwidth]{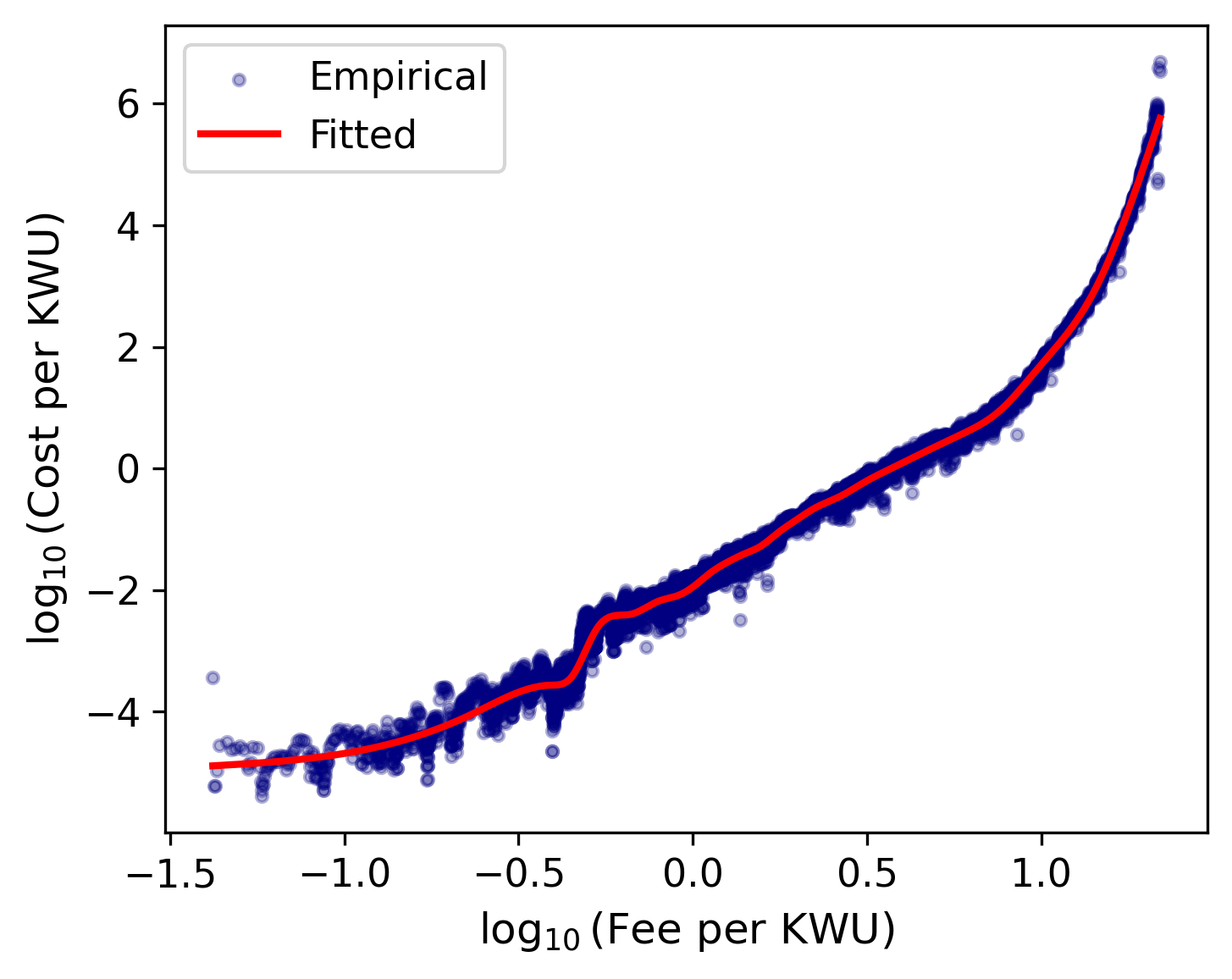}
        \caption{Cost-bid relationship}
        \label{fig:btc_costbid}
    \end{subfigure}
    \hfill
    \begin{subfigure}[b]{0.48\textwidth}
        \includegraphics[width=\textwidth]{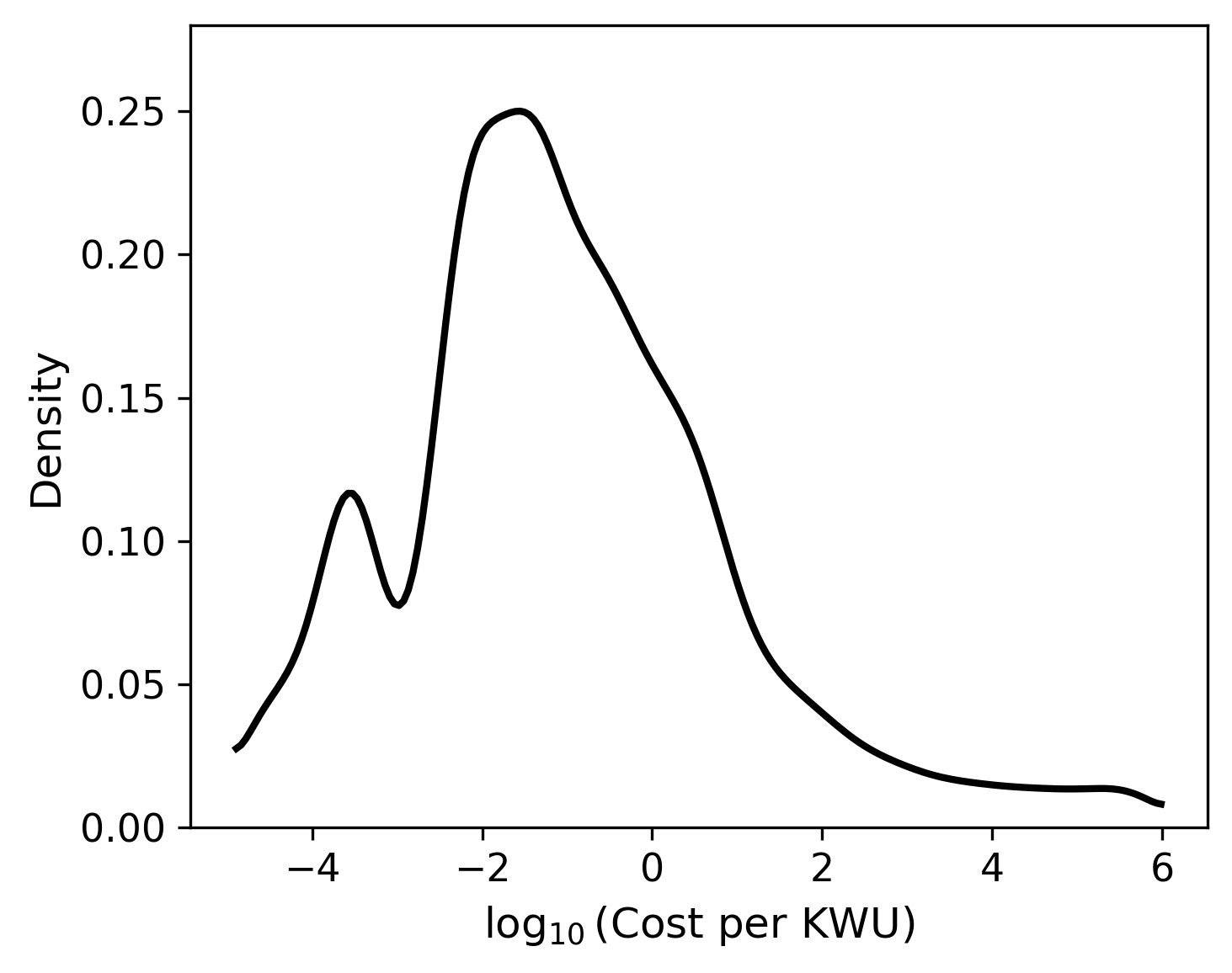}
        \caption{Cost distribution}
        \label{fig:btc_cost_dist}
    \end{subfigure}
    \hfill 
    \begin{subfigure}[b]{0.48\textwidth}
        \includegraphics[width=\textwidth]{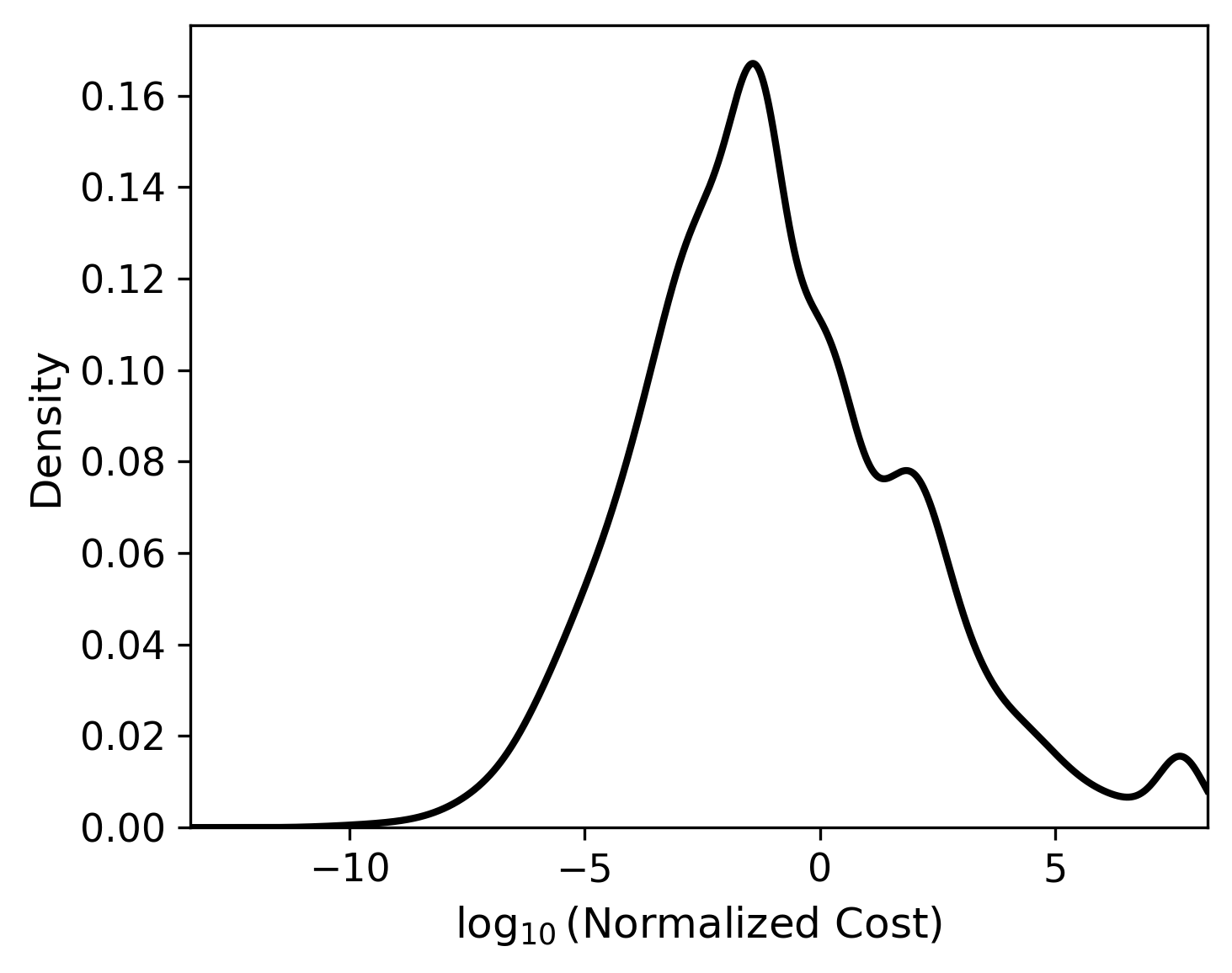}
        \caption{Normalized cost distribution}
        \label{fig:btc_cost_dist_all}
    \end{subfigure}
       \caption{The cost structure of Bitcoin users}
    \label{fig:btc_cost_analysis}
  \end{figure}

As explained at the beginning of this section, we directly estimate the bid-cost relationship using Equation~\eqref{eq:cost-bid}. Recall $\bW = \bW(\rho_i)$ for the $i$-th user. Also recall that the load parameter $\rho_i$ at equilibrium satisfies $\rho_i = \brho(c_i) = \bar\rho  \bar F(c_i) = \bar\rho \bar G(\bb(c_i))$. Given the actual fees and their empirical distribution function, we can apply the finite difference method to the left hand side of \eqref{eq:cost-bid} using the formula of $\bW$ by bumping $b_i$ (and thus $\rho_i$) up and down. Figure~\ref{fig:btc_costbid} presents the empirical relationship between fees and costs. More specifically, blue dots are pairs of an actual fee and the corresponding cost estimate. In order to ensure the monotonicity of cost versus fee (or vice versa), we apply the cubic spline method by suitably discretizing the $x$-axis. The red curve in the figure shows the fitted curve.  It clearly shows nonlinearity with costs sharply rising for larger fees. Figure~\ref{fig:btc_cost_dist} gives the empirical density of estimated costs. Note costs are in USD per KWU and per minute. Due to the bimodal feature of the fee distribution in Figure~\ref{fig:btc_feedist}, the cost distribution has two peaks as well at $10^{-4}$ and $10^{-2}$ approximately. 

To better reflect users' waiting cost and to make the cross-chain analysis useful, we introduce the concept of normalized cost by taking into account transaction specific attributes. Each transaction request is assigned a weight which captures computational and storage requirements (other networks in our paper use the storage requirement only). The transfer amount of each request may impact the bid because users are presumably sensitive to the value of the transaction. Lastly, we also adjust costs by expected block time because block generation times are different across chains. As a result, we apply
$$
c_{\sf normalized} = \frac{\text{expected block time}}{\text{transfer amount} \times \text{weight}} \times c.
$$
This standardization measures the user cost per USD transferred per weight or size during a single block generation time. The resulting distribution in Figure~\ref{fig:btc_cost_dist_all} shows a markedly different pattern from the previous cost structure. The normalized cost distribution exhibits a unimodal shape with its peak near 0.1, spanning a wider range from -10 to 8 in log scale. This normalization demonstrates that the bimodality in the distribution of cost per KWU emerges from transaction-specific parameters rather than underlying user preferences.

\begin{figure}[!ht]
\centering 
  \begin{subfigure}[b]{0.48\textwidth}
        \includegraphics[width=\textwidth]{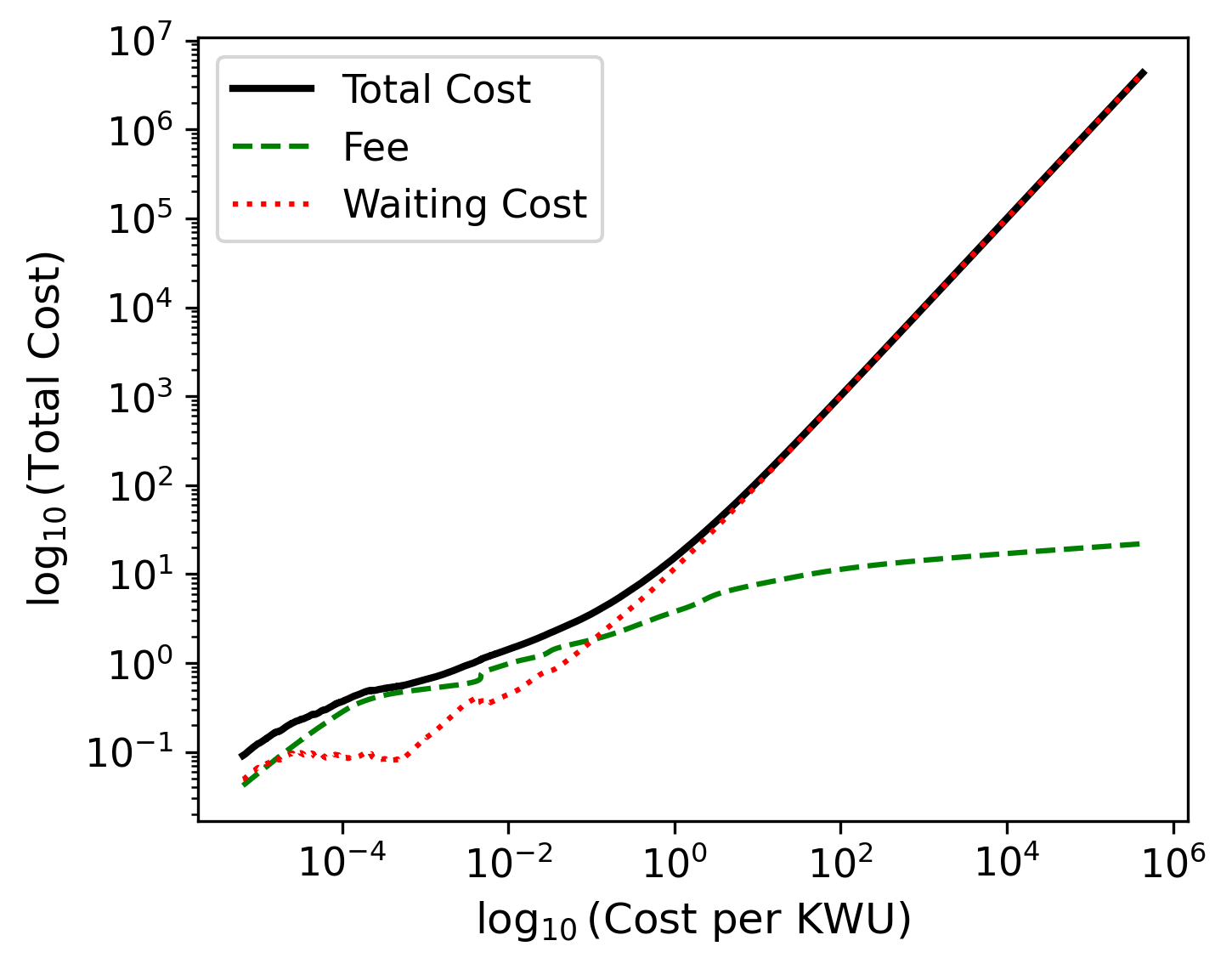}
        \caption{Total cost decomposition}
        \label{fig:btc_total_cost}
    \end{subfigure}
    \begin{subfigure}[b]{0.48\textwidth}
        \includegraphics[width=\textwidth]{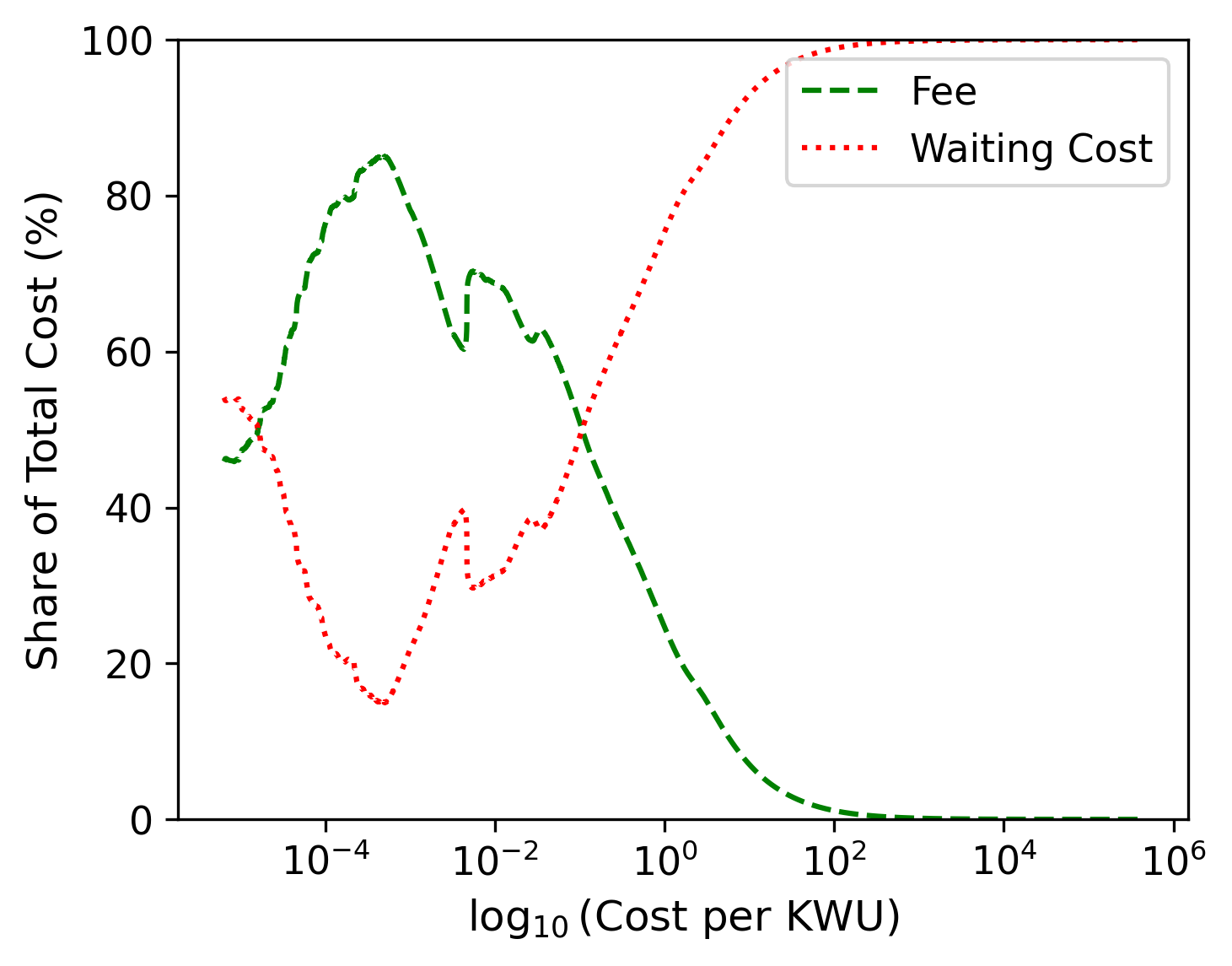}
        \caption{Cost ratio analysis}
        \label{fig:btc_cost_ratio}
    \end{subfigure}
    \caption{Analysis of total cost}
    \label{fig:btc_cost_all}
\end{figure}

Users' total costs are composed of this explicit transaction fee and the expected waiting cost, that is,  $b + c \bW$. As seen in Figure~\ref{fig:btc_total_cost}, we see the relative contribution of fee and waiting cost to the total cost. When a user has a high cost greater than 100 USD/KWU$\cdot$min, the waiting cost is a dominant factor whereas fees are dominant for low cost users (less than 0.01 USD/KWU$\cdot$min). See Figure~\ref{fig:btc_cost_ratio}. One possible explanation is that high cost users optimize their transaction requests by strategic bidding.


\subsection{Optimal Response in Bidding}

Having established the empirical cost structure $F$, we are now able to see how our model yields optimal bids if there is any change in block distribution $F_B$. The baseline scenario of the Bitcoin network is that the block distribution is exponential with mean 10 minutes. Let us consider different block distributions with the same mean. More specifically, we use the gamma distribution $\Gamma(\alpha, \alpha/10)$ parameterized by shape parameter $\alpha$ and rate parameter $\alpha/10$. The mean stays at 10 but the variance $100/\alpha$. We choose three cases: high variance ($\alpha=0.2$, variance=500), moderate variance ($\alpha=1$ or exponential, variance=100), and low variance ($\alpha=5$, variance=20). The block distribution converges to a constant block distribution time as $\alpha \rightarrow \infty$, which is included as a limiting case.  

\begin{figure}[!ht]
    \centering
    \begin{subfigure}[b]{0.48\textwidth}
    \includegraphics[width=\textwidth]{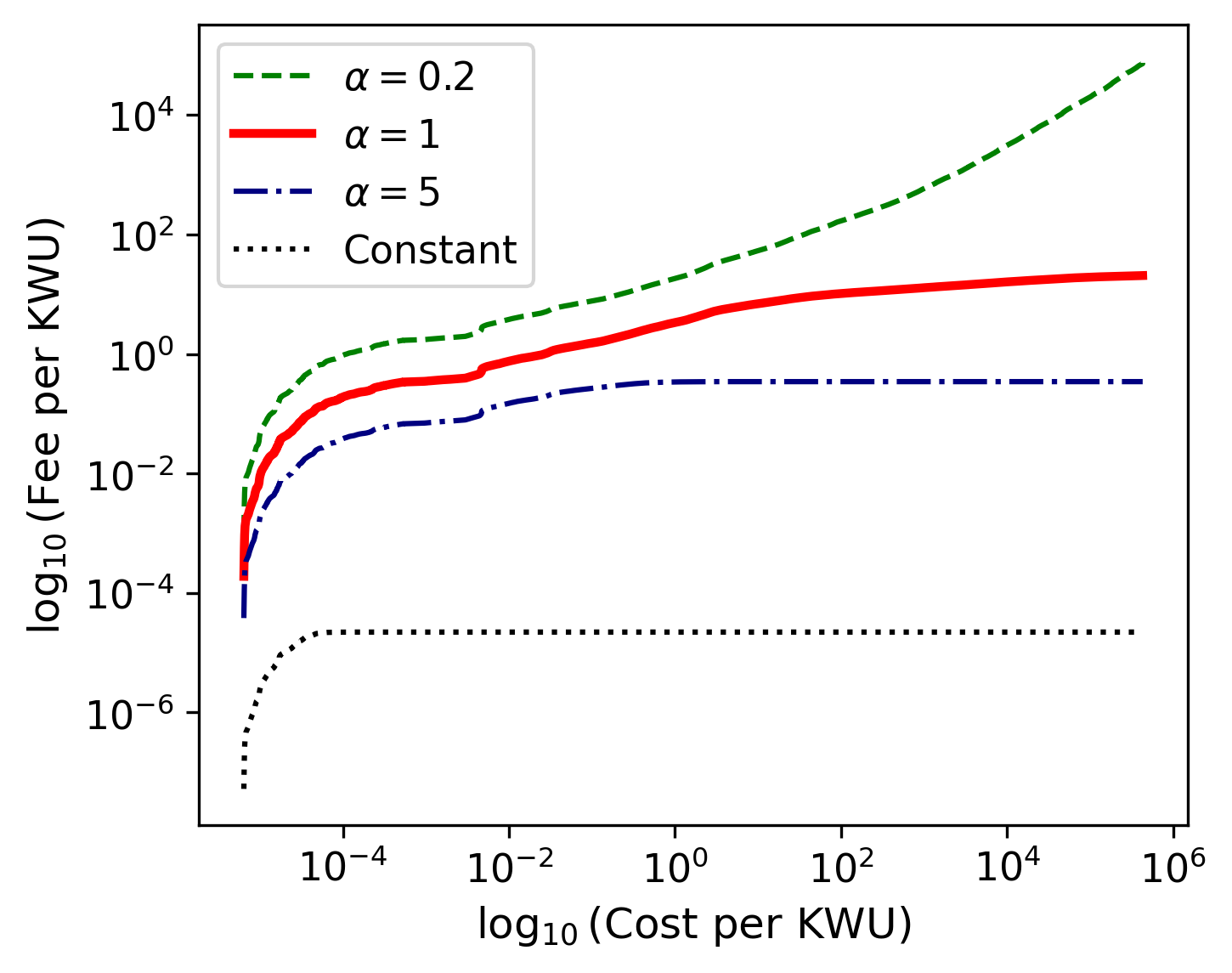}
    \caption{Cost-bid relationship}
    \label{fig:gamma_costbid}
    \end{subfigure}
    \hfill
    \begin{subfigure}[b]{0.48\textwidth}
    \includegraphics[width=\textwidth]{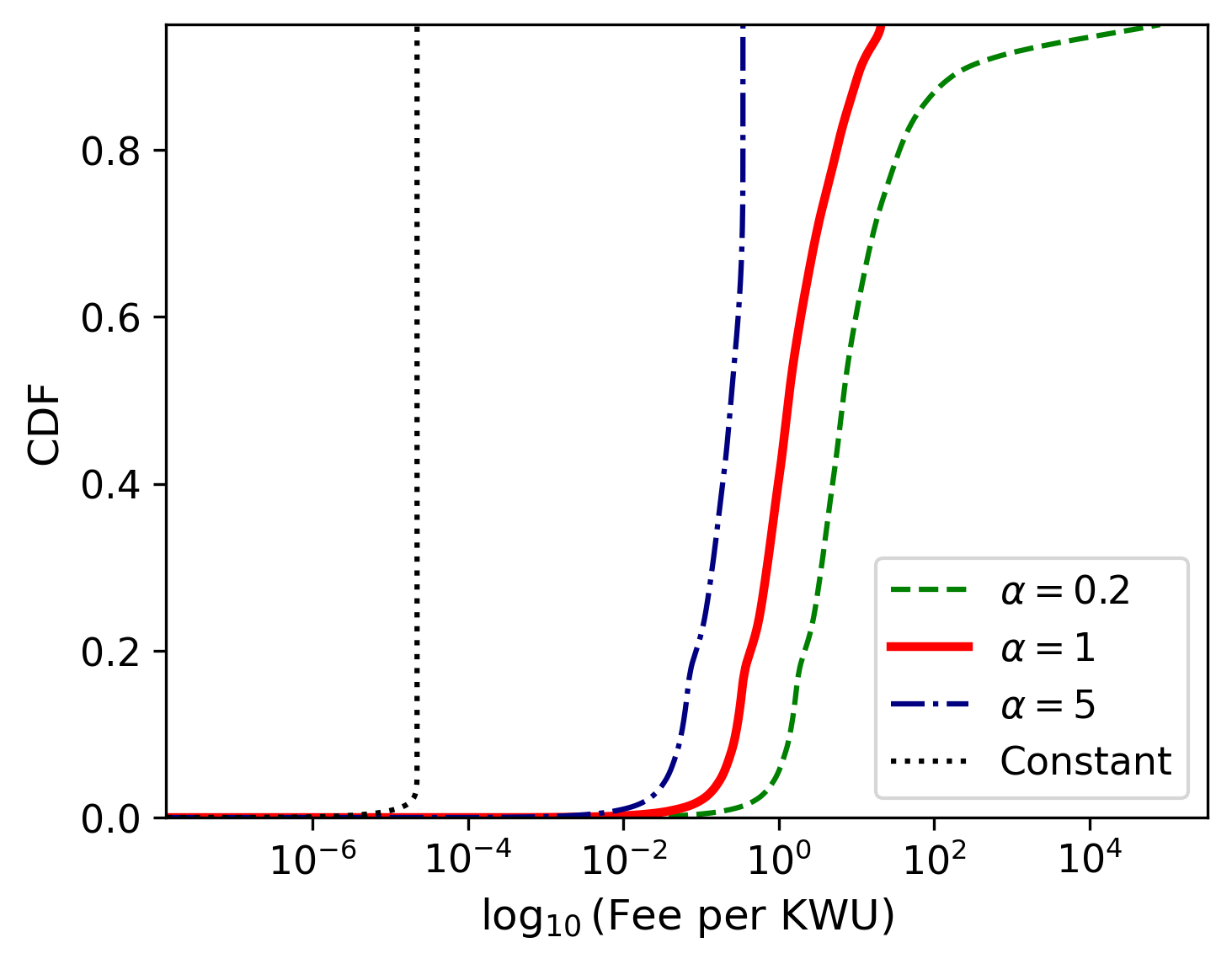}
    \caption{Bid distribution}
    \label{fig:gamma_bidcdf}
    \end{subfigure}
    \hfill
    \begin{subfigure}[b]{0.48\textwidth}
    \includegraphics[width=\textwidth]{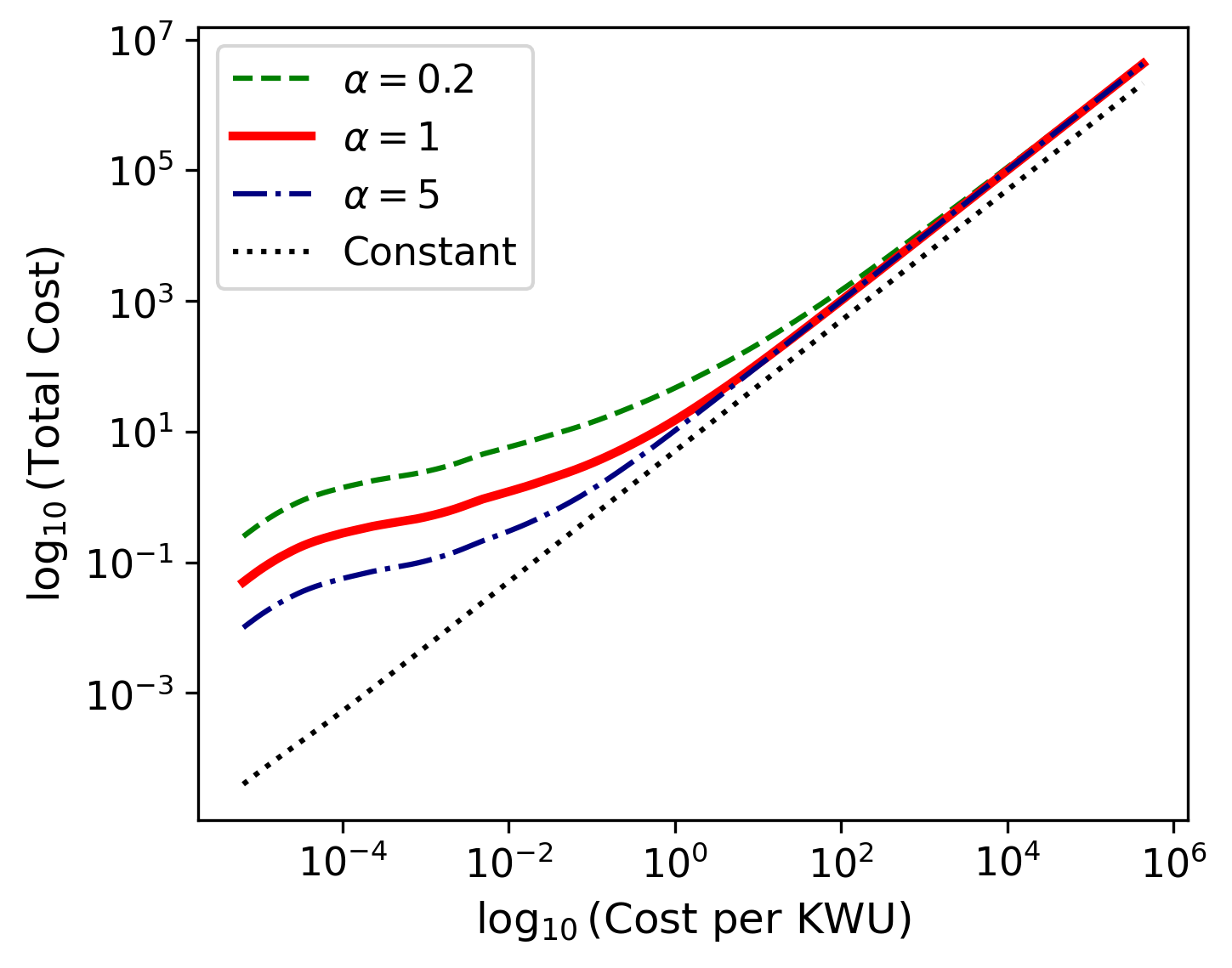}
    \caption{Comparison of total costs}
    \label{fig:gamma_totalcost}
    \end{subfigure}
    \caption{Impact of block time distribution on transaction costs}
    \label{fig:gamma_analysis}
\end{figure}

Figure~\ref{fig:gamma_analysis} shows three computational results. From Figure~\ref{fig:gamma_costbid}, we see optimal bidding strategies in terms of cost. The high variance case ($\alpha=0.2$) exhibits significantly higher bids across all cost levels, particularly for users with high delay costs. This suggests that the increased uncertainty in block generation time makes users bid higher to minimize the total cost. In the limiting case, users make the lowest and mostly constant bid at different cost levels. This is a predictable outcome because users can precisely predict transaction confirmation times in this extreme case. This decreasing pattern of variability in bidding as $\alpha$ increases is also found in the cumulative distribution of bids in Figure~\ref{fig:gamma_bidcdf}.  The observation that higher variability in $B$ results in larger bids and thus to larger total costs is re-confirmed in the last panel. 

The analysis reveals that block time uncertainty significantly impacts optimal bidding strategies, with higher variance leading to more aggressive bidding behavior. However, it is also seen that such responses to variability in block distribution are disproportionate at different cost levels. In order to understand such patterns in response better, we examine two other scenarios. First, we examine the impact of mining delay, which might happen due to optimal strategic choices of miners \citep{carlsten2016instability, kim2024mind}. Second, the impact of increased average block generation time is studied. This might occur due to difficulty attacks as seen in the Bitcoin vs Bitcoin Cash hash power competition or external shocks (e.g. mining restrictions in major mining regions). 

\begin{figure}[!ht]
    \centering
    \begin{subfigure}[b]{0.48\textwidth}
    \includegraphics[width=\textwidth]{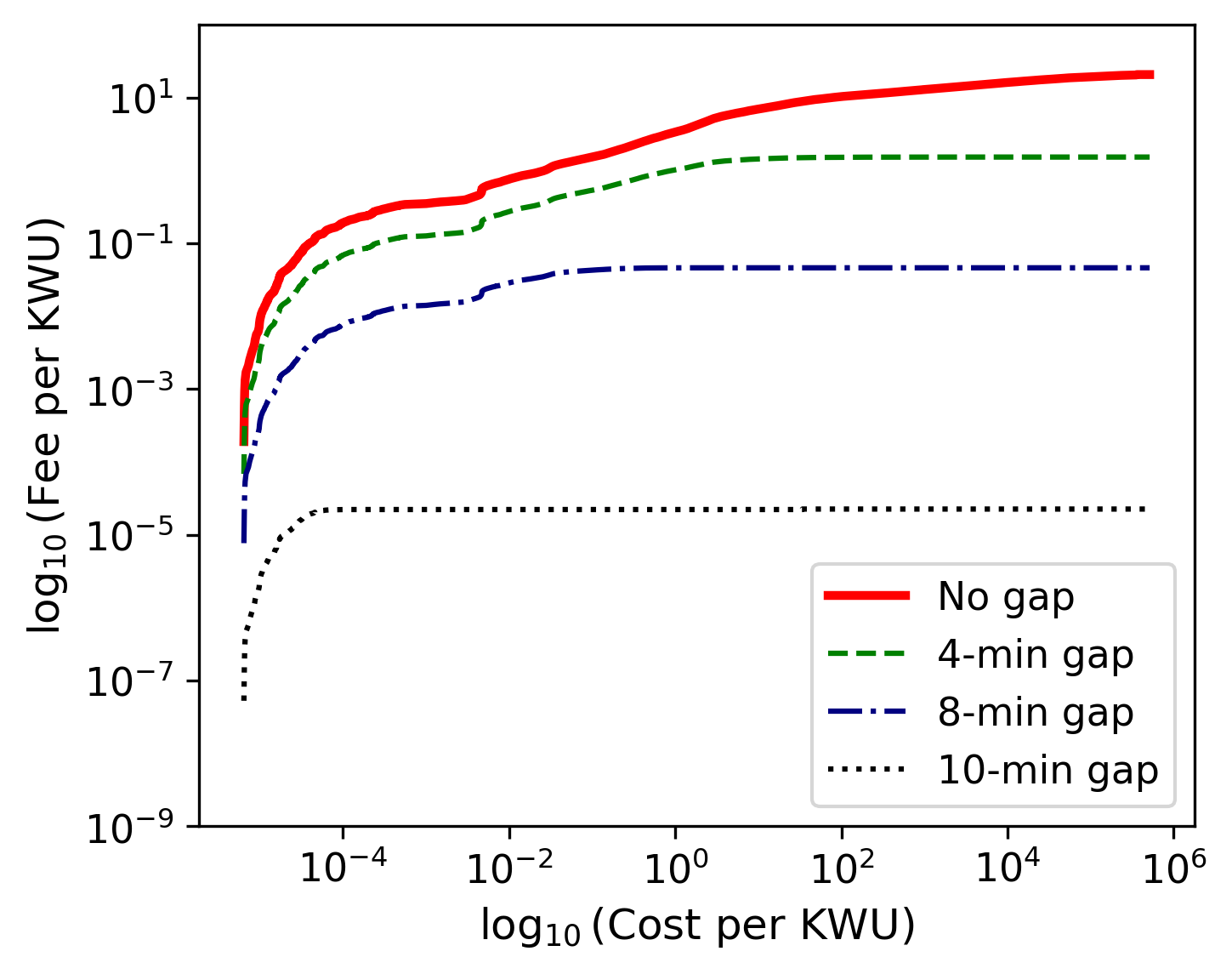}
    \caption{Cost-bid relationship}
    \label{fig:gap_costbid}
    \end{subfigure}
    \hfill
    \begin{subfigure}[b]{0.48\textwidth}
    \includegraphics[width=\textwidth]{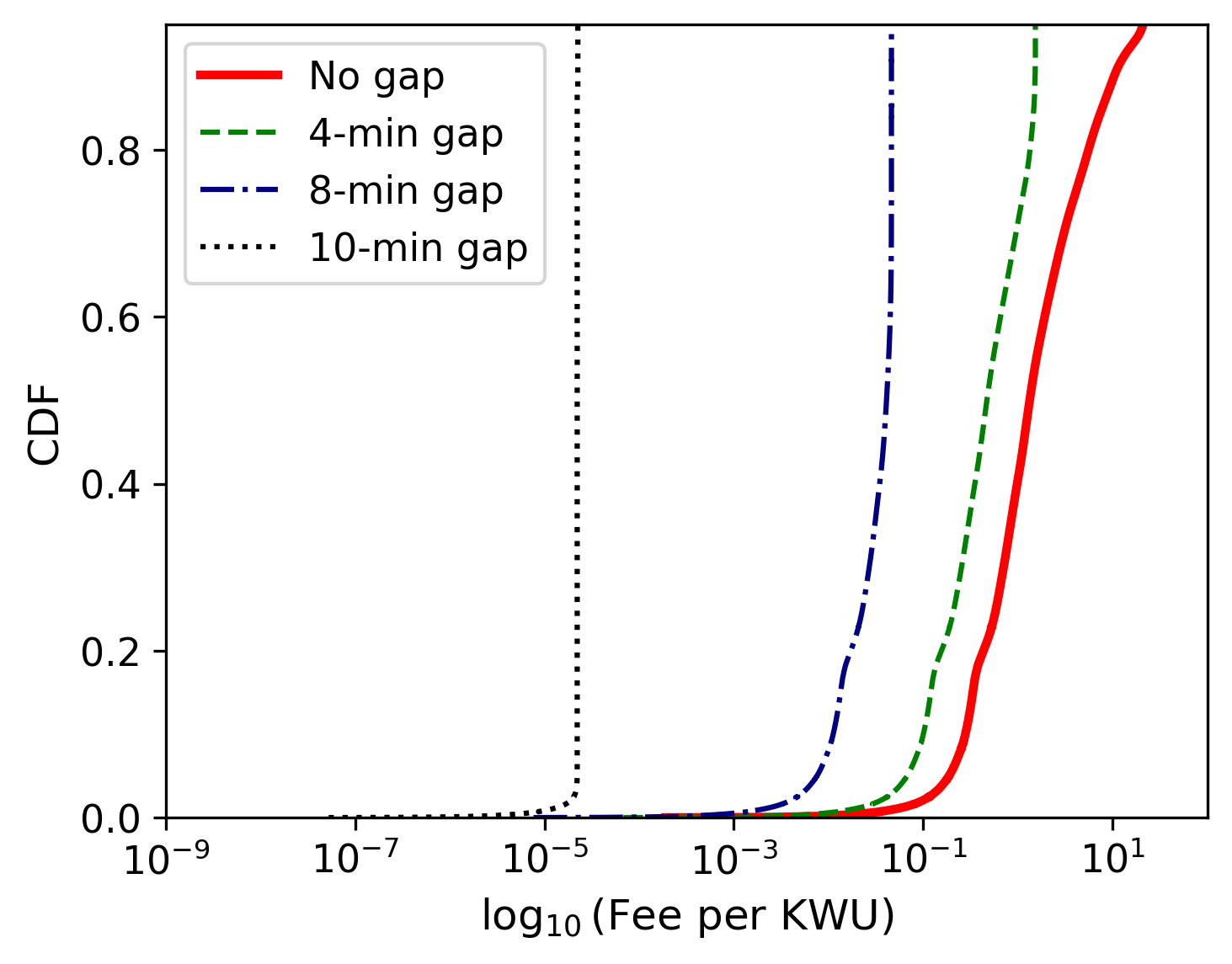}
    \caption{Bid distribution}
    \label{fig:gap_bidcdf}
    \end{subfigure}
    \caption{Impact of mining gaps on transaction costs}
    \label{fig:gap_analysis}
\end{figure}

The first scenario is implemented by imposing constant delays (baseline or no gap, 4 minutes, 8 minutes, and 10 minutes) in block generation. The results are shown in Figure~\ref{fig:gap_analysis}. The left panel shows a systematic decline in bidding as the mining gap increases. With no mining gap, the expected fee per KWU is 2.87 USD, which drops to 0.61 USD with a 4-minute gap and further decreases to 0.033 USD with an 8-minute gap. This decline reflects the reduced willingness of users to pay high fees for static delays. The variance in bids is also significantly reduced as shown in the right panel. This indicates that mining gaps lead to more predictable and compressed fee markets, potentially harming the effectiveness of fee-based transaction prioritization. Considering that mining gaps may occur due to insufficient mining rewards, this reduced fee may exacerbate the situation. 

\begin{figure}[!ht]
    \centering
    \begin{subfigure}[b]{0.48\textwidth}
        \includegraphics[width=\textwidth]{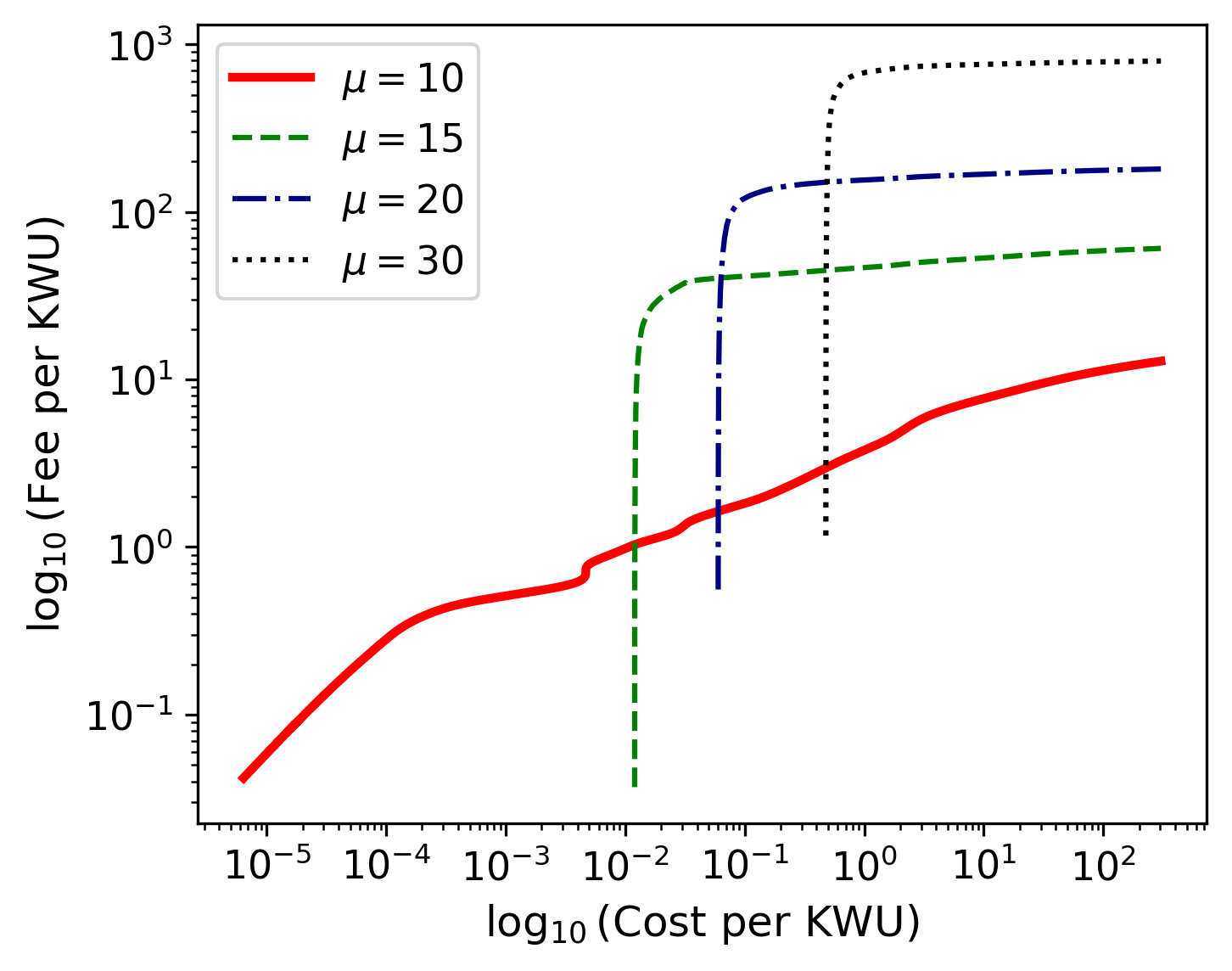}
        \caption{Cost-bid relationship}
        \label{fig:delay_mu_costbid}
    \end{subfigure}
    \hfill
    \begin{subfigure}[b]{0.48\textwidth}
        \includegraphics[width=\textwidth]{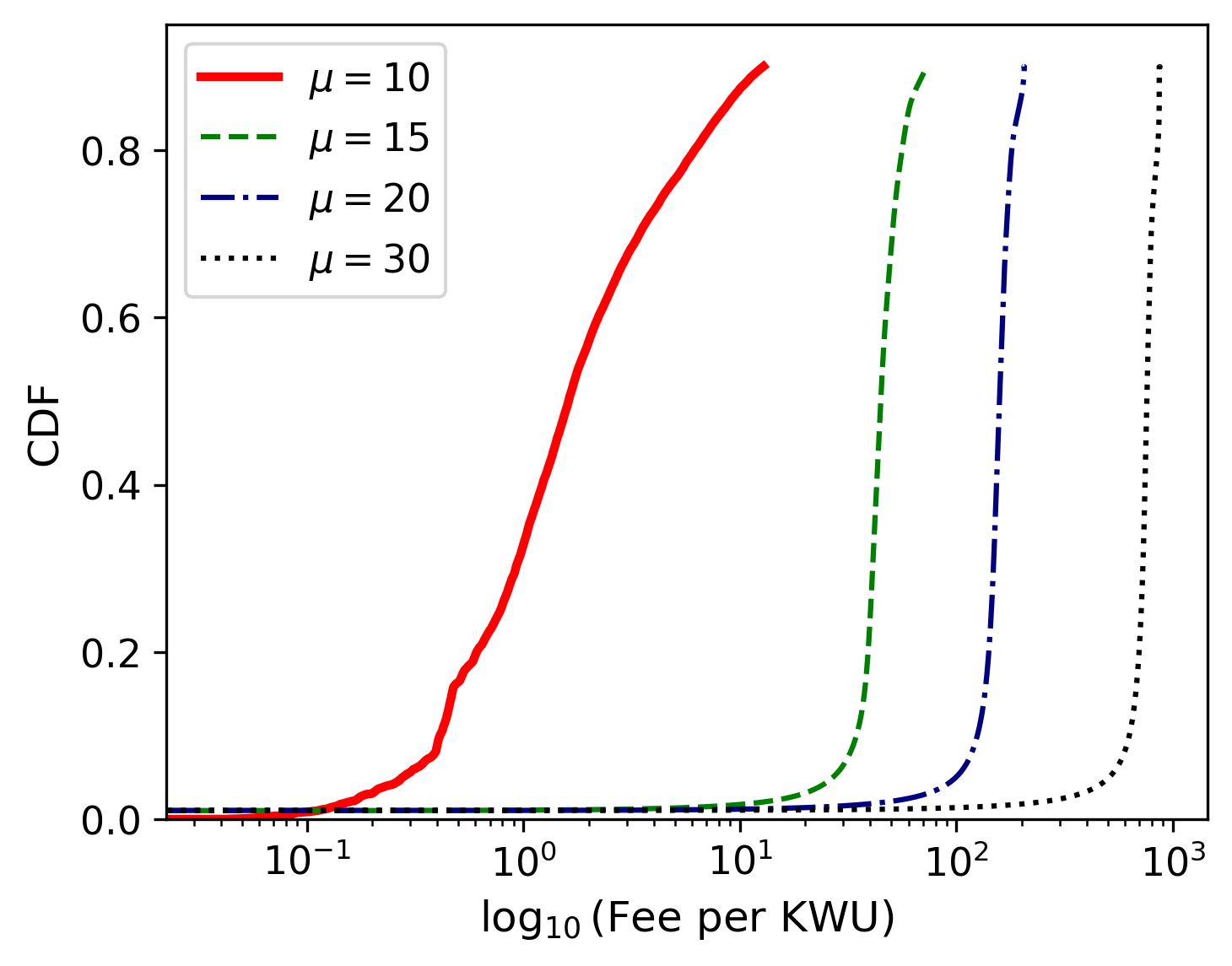}
        \caption{Bid distribution}
        \label{fig:delay_mu_bidcdf}
    \end{subfigure}
    \caption{Impact of extended block times on transaction costs}
    \label{fig:delay_mu_analysis}
\end{figure}

The second scenario is implemented by varying the mean parameter $\mu$ of the exponential distribution while maintaining its memoryless property. Figure~\ref{fig:delay_mu_analysis} delivers the same information of cost-bid relationship and bid distribution as in previous figures. The left panel demonstrates that longer block times lead to systematically higher optimal bids across all cost  levels. As $\mu$ increases from 10 (baseline case) to 30 minutes, the mean bid rises from 2.45 USD  to 732.93 USD. This substantial increase is in stark contrast of the significant decrease in the first scenario where users experience a fixed delay. The right panel can also be compared with Figure~\ref{fig:gap_bidcdf}. Optimal bidding becomes more homogeneous as the mean block generation time increases, and the coefficient of variation changes from 1.087 in the baseline case to 0.155 at $\mu=30$. 

These findings have implications for blockchain security and protocol design. Different block distributions have different impacts on optimal bidding. When there is a larger uncertainty or larger mean waiting time, bid levels are higher. This could impose substantial costs on network users in terms of waiting time and transaction fee,   potentially threatening the network's utility as a payment system. When there is a mining gap, on the other hand, bid levels are lower but this might lead to insufficient economic rewards for miners. 
Lastly, the increased homogeneity in bidding behavior under large mining gaps or extended block generation times indicates a potential reduction in the fee market's ability to effectively prioritize transactions based on urgency.

\subsection{Cross Chain Analysis}

\begin{figure}[!ht]
    \centering
    \begin{subfigure}[b]{0.48\textwidth}
        \includegraphics[width=\textwidth]{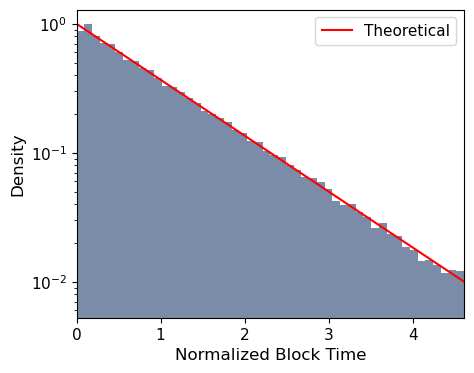}
        \caption{LTC block time distribution}
        \label{fig:ltc_dist}
    \end{subfigure}
    \hfill
    \begin{subfigure}[b]{0.48\textwidth}
        \includegraphics[width=\textwidth]{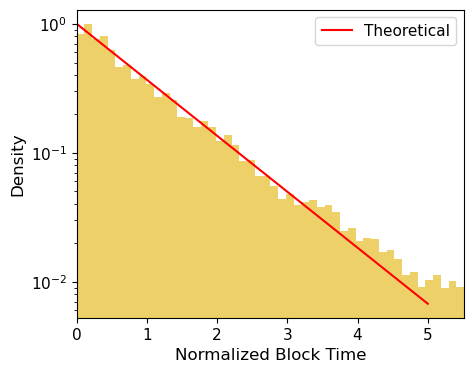}
        \caption{DOGE block time distribution}
        \label{fig:doge_dist}
    \end{subfigure}
    \hfill
    \begin{subfigure}[b]{0.48\textwidth}
        \includegraphics[width=\textwidth]{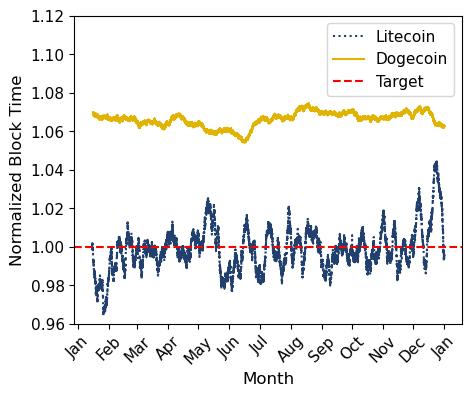}
        \caption{Block time stability comparison}
        \label{fig:blocktime_stability}
    \end{subfigure}
    \caption{Block time characteristics across networks}
    \label{fig:blocktime_analysis}
\end{figure}

We extend our analysis to other major Proof-of-Work cryptocurrencies to examine how different network parameters affect fee market dynamics. We focus on Litecoin (LTC), and Dogecoin (DOGE), analyzing their transaction data throughout the year 2023. These networks have different design parameters: LTC targets 2.5-minute blocks with $K$=1,384 and DOGE aims for 1-minute blocks with $K$=2,597. The corresponding utilization rates ($\bar{\rho}$) are 0.2291 and 0.1076 respectively.
The block time distributions (Figures~\ref{fig:ltc_dist}-\ref{fig:doge_dist}) closely follow exponential distributions when normalized by their respective target times, confirming the memoryless property of the mining process across different networks. Figure~\ref{fig:blocktime_stability} illustrates the stability of the average block time over the course of 2023. In contrast, Dogecoin consistently exhibits a deviation of approximately 7\%. This is presumed to be due to the fact that the target block time is as short as one minute, and the physical network propagation time (approximately three seconds) has a greater impact. 

\begin{figure}[!ht]
    \centering
    \begin{subfigure}[b]{0.48\textwidth}
    \includegraphics[width=\textwidth]{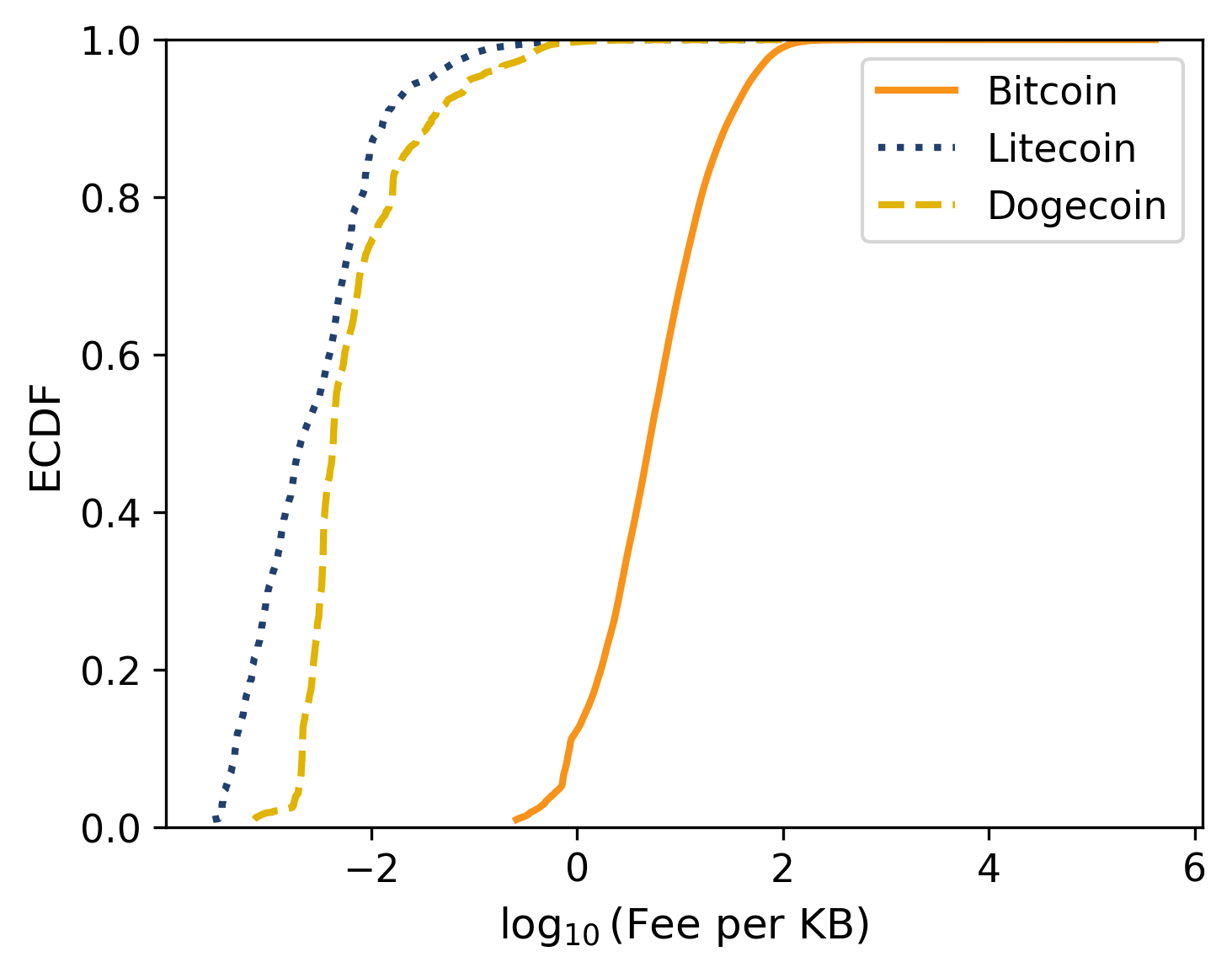}
    \caption{Fee ECDF comparison}
    \label{fig:fee_ecdf}
    \end{subfigure}
    \hfill
    \begin{subfigure}[b]{0.48\textwidth}
    \includegraphics[width=\textwidth]{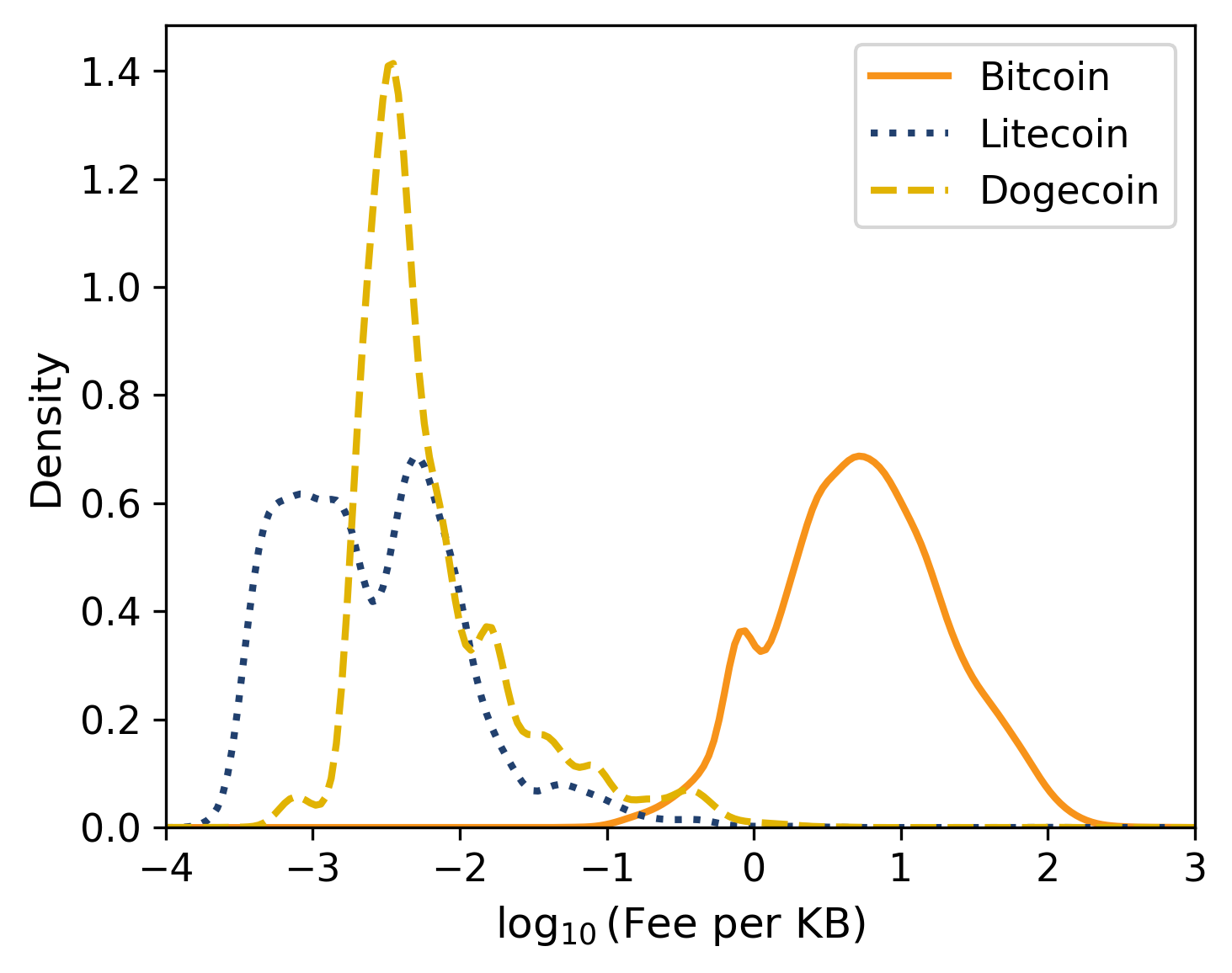}
    \caption{Fee density comparison}
    \label{fig:fee_dist}
    \end{subfigure}
    \caption{Fee distributions across networks}
    \label{fig:fee_analysis}
\end{figure}

The fee structures exhibit considerable heterogeneity across networks, stemming from their distinct transaction architectures and prioritization mechanisms. While Bitcoin employs weight units for transaction prioritization, Litecoin and Dogecoin utilize physical size. Despite these differences, fee per KB metrics are available across all networks, enabling cross-chain comparison as shown in Figure~\ref{fig:fee_ecdf}. The empirical cumulative distribution functions and density plots reveal markedly different fee patterns across these networks.
The density distributions demonstrate substantial divergence in fee structures, with no apparent commonalities across chains. Bitcoin exhibits a relatively concentrated distribution with a primary peak around $10^1$ fee-per-KB, while Litecoin shows a bimodal pattern with peaks at approximately $10^{-3}$ and $10^{-2}$. Dogecoin displays the most distinctive pattern with a sharp, concentrated peak at $10^{-2.5}$. 

\begin{figure}[!ht]
    \centering
    \includegraphics[width=0.6\textwidth]{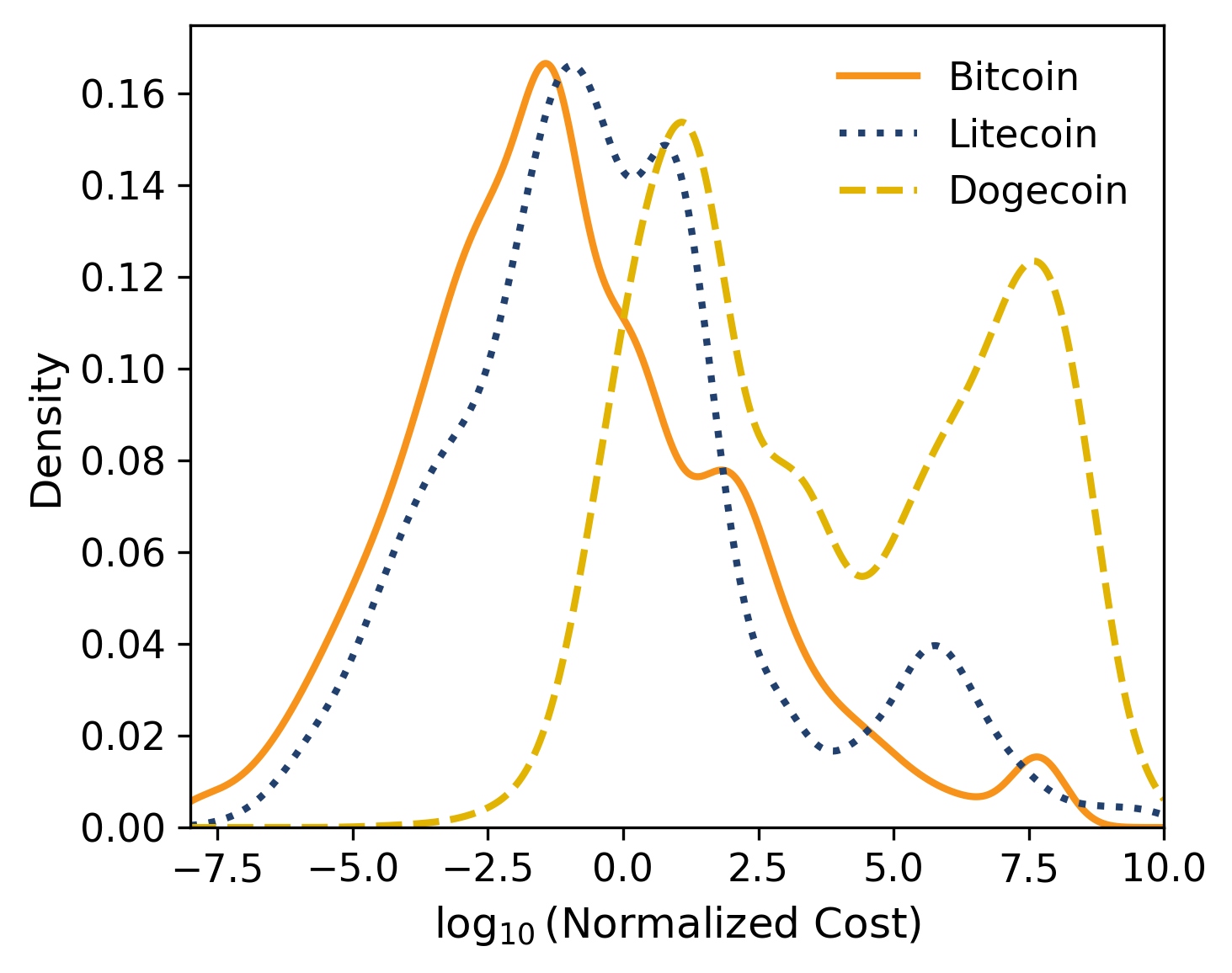}
    \caption{Normalized cost structures across networks}
    \label{fig:cost_pdf}
\end{figure}

We apply the same procedure as done for the Bitcoin network to retrieve the hidden cost structure. The normalized cost distributions in Figure~\ref{fig:cost_pdf} reveal similarities in the underlying cost structures across networks, despite their different fee patterns. All three networks show a primary concentration of user activity around normalized costs of $10^{-1}$ to $10^1$, suggesting common patterns in user behavior when transaction-specific attributes are accounted for. 
While sharing fundamental similarities, each network exhibits distinctive characteristics in their detailed structure. Notably, Dogecoin displays an additional mode around $10^7$ in the normalized cost distribution, attributable to its minimum fee requirements (0.01 DOGE dust limit, 0.001 DOGE/kB) that might affect small-value transactions. Litecoin similarly shows a secondary peak at higher cost levels, reflecting the impact of minimum fee constraints rather than queue congestion.

\section{Conclusion}\label{sec:conclusion}

This paper has investigated transaction fees and their implications in decentralized networks, particularly Proof-of-Work based blockchain systems. Through our queueing based modeling and analysis, we conducted an analysis of a priority queue with bulk service and applied the result to understand strategic user behaviors in bidding transaction fees. 
Central to our investigation, more specifically, was the adoption of the M/G\textsuperscript{$K$}/1 priority queue model to represent the blockchain queue. This study contributes to the literature on the queueing approach to blockchain systems in that, first, a concrete semi-closed form expressions for steady-state quantities such as expected queue length are presented and, second, the hidden delay cost structure of blockchain users is presented when the block generation time has a general distribution. 

To be more specific, we find that estimating waiting time in steady-state distribution is instrumental in pinpointing optimal fees for users across various queue scenarios. Our analysis of normalized transaction costs across different blockchain networks reveals remarkably similar underlying cost structures, suggesting common patterns in user behavior despite varying network parameters and congestion levels. This finding provides empirical supports for our theoretical framework linking delay costs to optimal bidding strategies. Furthermore, our approach enables us to predict how users would react to a change in blockchain characteristics. We demonstrated that users would bid higher fees if there is a higher uncertainty in block generation time but bid less if there is a fixed gap in block generations. 

Nevertheless, this work has limitations. One such limitation is our model's exclusion of post-transaction fee modification strategies like Child-Pays-for-Parent  and Replace-by-Fee. These strategies allow users to elevate their transactions' priority post submission, introducing additional layers of complexity to queue dynamics that our current model does not address.
%
%
Notwithstanding these limitations, we believe this study advances our understanding of the intricate relationship between block generations, transaction fee dynamics, and the  hidden delay cost structure of users. The insights from the study should be helpful for refining fee mechanisms and for enhancing efficiency and fairness on blockchains. 
%

\section*{Acknowledgement}
    This work was supported by the National Research Foundation of Korea (NRF) funded by the Korea Government (MSIT) under Grant RS-2023-00278082 and by the Institute of Information \& Communications Technology Planning \& Evaluation (IITP)‐Global Data‐X Leader HRD program grant funded by the Korea government (MSIT) (IITP‐RS‐2024‐00440626).  

\section*{Disclosure of Interest}

There are no relevant financial or non-financial competing interests to report.


\section*{Disclaimer}

The work of Donghwa Seo was completed prior to joining Samsung SDS. The views and conclusions expressed herein are those of the authors and do not necessarily reflect any policy or position of Samsung SDS. 

\bibliographystyle{plainnat}
\bibliography{Ch3_ref}

@article{kim2024mind,
  title={Mind the gap in the mining game},
  author={Kim, K. and Seo, D.},
  year = {2024},
  note={Working Paper},
  journal={}
}

@book{chaudhry1983first,
  title={A First Course in Bulk Queues},
  author={Chaudhry, M. L. and Templeton, J. G. C.},
  publisher={Wiley},
  year={1983}
}

@article{downton1956limiting,
  title={On limiting distributions arising in bulk service queues},
  author={Downton, F.},
  journal={Journal of the Royal Statistical Society Series B: Statistical Methodology},
  volume={18},
  number={2},
  pages={265--274},
  year={1956}
}

@article{bailey1954queueing,
  title={On queueing processes with bulk service},
  author={Bailey, N. T. J.},
  journal={Journal of the Royal Statistical Society: Series B (Methodological)},
  volume={16},
  number={1},
  pages={80--87},
  year={1954}
}

@article{oblakova2019exact,
  title={An exact root-free method for the expected queue length for a class of discrete-time queueing systems},
  author={Oblakova, A. and Al Hanbali, A. and Boucherie, R. J. and van Ommeren, J. C. W. and Zijm, W. H. M.},
  journal={Queueing Systems},
  volume={92},
  pages={257--292},
  year={2019}
}

@inproceedings{carlsten2016instability,
author = {Carlsten, M. and Kalodner, H. and Weinberg, S. M. and Narayanan, A.},
title = {On the Instability of {B}itcoin Without the Block Reward},
year = {2016},
booktitle = {Proceedings of the 2016 ACM SIGSAC Conference on Computer and Communications Security},
pages = {154–-167}
}

@inproceedings{kawase2017transaction,
  title={Transaction-confirmation time for {B}itcoin: {A} queueing analytical approach to blockchain mechanism},
  author={Kawase, Y. and Kasahara, S.},
  editor={Yue, W. and Li, Q.-L. and  Jin, S. and Ma, Z.},
  booktitle={Queueing Theory and Network Applications},
  pages={75--88},
  year={2017}
}

@article{huberman2021monopoly,
  title={Monopoly without a monopolist: {A}n economic analysis of the {B}itcoin payment system},
  author={Huberman, G. and Leshno, J. D. and Moallemi, C.},
  journal={Review of Economic Studies},
  volume={88},
  number ={6},
  pages={3011--3040},
  year={2021}
}

@article{easley2019mining,
  title={From mining to markets: {T}he evolution of bitcoin transaction fees},
  author={Easley, D. and O'Hara, M. and Basu, S.},
  journal={Journal of Financial Economics},
  volume={134},
  number={1},
  pages={91--109},
  year={2019}
}

@InProceedings{li2018blockchain,
author="Li, Q.-L. and Ma, J.-Y. and Chang, Y.-X.",
editor="Chen, X. and Sen, A. and Li, W. W. and Thai, M. T.",
title="Blockchain queue theory",
booktitle="The 7th International Conference on Computational Data and Social Networks",
year="2018",
pages="25--40"
}

@inproceedings{papadis2018stochastic,
  title={Stochastic models and wide-area network measurements for blockchain design and analysis},
  author={Papadis, N. and Borst, S. and Walid, A. and Grissa, M. and Tassiulas, L.},
  booktitle={IEEE Conference on Computer Communications},
  pages={2546--2554},
  year={2018}
}

@article{kasahara2019effect,
  title={Effect of {B}itcoin fee on transaction-confirmation process},
  author={Kasahara, S. and Kawahara, J.},
  journal={Journal of Industrial \& Management Optimization},
  volume={15},
  number={1},
  pages={365},
  year={2019},
  publisher={American Institute of Mathematical Sciences}
}

@article{banerjee2015analysis,
  title={Analysis of a finite-buffer bulk-service queue with multiple servers and batch {M}arkovian arrival process},
  author={Banerjee, A. and Gupta, U. C. and Chakravarthy, S. R.},
  journal={Computers \& Operations Research},
  volume={60},
  pages={138--149},
  year={2015}
}

@article{claeys2013tail,
  title={Tail probabilities of the delay in a batch-service queueing system with batch-size dependent service times and a timer mechanism},
  author={Claeys, D. and Steyaert, B. and Walraevens, J. and Laevens, K. and Bruneel, H.},
  journal={Computers \& Operations Research},
  volume = {40},
  number ={5},
  pages={1497--1505},
  year={2013}
}

@book{hassin2003queue,
  title={To Queue or Not To Queue: {E}quilibrium Behavior in Queueing Systems},
  author={Hassin, R. and Haviv, M.},
  year={2003},
  publisher={Springer}
}

@article{dimitri2019transaction,
  title={Transaction fees, block size limit, and auctions in {B}itcoin},
  author={Dimitri, N.},
  journal={Ledger},
  volume={4},
  year={2019}
}

@inproceedings{moser2015trends,
  title={Trends, tips, tolls: {A} longitudinal study of {B}itcoin transaction fees},
  author={M{\"o}ser, M. and B{\"o}hme, R.},
  booktitle={International Conference on Financial Cryptography and Data Security},
  editors ={Brenner, M. and Christin, N. and Johnson, B. and Rohloff, K.},
  pages={19--33},
  year={2015}
}

@article{houy2014economics,
  title={The economics of {B}itcoin transaction fees},
  author={Houy, N.},  
  journal = {},
  year={2014},
 note={Working Paper}
}

@article{ilk2021stability,
   title={Stability of Transaction Fees in {B}itcoin: {A} Supply and Demand Perspective},
   author = {Ilk, N. and Shang, G. and Fan, S. and Zhao, J. L.}, 
   journal = {MIS Quarterly},
   volume = {45},
   number={2},
   year={2021},
   pages ={563--592}
}

@article{lavi2022redesign,
author = {Lavi, R. and Sattath, O. and Zohar, A.},
title = {Redesigning {B}itcoin’s Fee Market},
year = {2022},
journal={ACM Transactions on Economics and Computation},
volume = {10},
number = {1},
pages={Article 5}
}

@article{basu2023stable,
author = {Basu, S. and Easley, D. and O’Hara, M. and Sirer, E. G.},
title={Stable{F}ees: {A} Predictable Fee Market for Cryptocurrencies},
year={2023},
volume={69},
number={11},
journal={Management Science},
pages={6508--6524}
}

@article{rough2024fee,
author = {Roughgarden, T.},
title = {Transaction Fee Mechanism Design},
year = {2024},
volume = {71},
number = {4},
journal = {Journal of the ACM},
pages={Article 30}
}
\end{document}